\documentclass[11pt,a4paper]{article}
\usepackage{amsmath,latexsym,amssymb,color,amsthm}
\usepackage{fullpage}
\usepackage{ifthen,graphics,epsfig}
\usepackage{xspace}
\usepackage{times}
\usepackage[english]{babel}
\usepackage{url,hyperref}
\usepackage{wrapfig}

\usepackage{float}
\newfloat{algorithm}{thp}{lop}
\floatname{algorithm}{Algorithm}
\newlength{\reduceunderfig}
\usepackage[lmargin=1.5cm,rmargin=2cm,tmargin=2cm,bmargin=2cm]{geometry}
\renewcommand{\baselinestretch}{0.96}
\renewcommand{\paragraph}[1]{\vspace{0.2cm}\noindent \textbf{#1}~~}
\let\bibitemOld\bibitem
\renewcommand{\bibitem}[1]{\vspace{-0.2cm} \bibitemOld{#1}}
\let\subsectionOld\subsection
\renewcommand{\subsection}[1]{\vspace{-0.2cm}\subsectionOld{#1}\vspace{-0.2cm}}
\let\subsubsectionOld\subsubsection
\renewcommand{\subsubsection}[1]{\vspace{-0.3cm}\subsubsectionOld{#1}\vspace{-0.15cm}}

\newcommand{\Xomit}[1]{}
\begin{document}
\newlength {\squarewidth}
\renewenvironment {square}
{
\setlength {\squarewidth} {\linewidth}
\addtolength {\squarewidth} {-12pt}
\renewcommand{\baselinestretch}{0.75} \footnotesize
\begin {center}
\begin {tabular} {|c|} \hline
\begin {minipage} {\squarewidth}
\medskip
}{
\end {minipage}
\\ \hline
\end{tabular}
\end{center}
}

\newcommand{\m}[1]{\ensuremath{\mathcal{#1}}\xspace}
\newcommand{\ID}{\mbox{\rm ID}}
\newcommand{\id}{\mbox{\rm id}}
\newcommand{\val}{\mbox{\sl val}}
\newcommand{\setofval}{\mbox{\sl set}}
\renewcommand{\dim}{\mbox{\sl dim}}
\newcommand{\carrier}{\mbox{\sl carrier}}
\newcommand{\complex}{\mbox{\sl complex}}
\newcommand{\skel}{\mbox{\rm skel}}

\newcommand{\OP}{\widehat{OP}}
\newcommand{\MEM}{\mathit{MEM}}

\newtheorem{definition}{Definition}
\newtheorem{theorem}{Theorem}
\newtheorem{lemma}{Lemma}
\newtheorem{corollary}{Corollary}
\newtheorem{claim}{Claim}
\newtheorem{property}{Property}
\newtheorem{remark}{Remark}

\newcommand{\toto}{xxx}
\newenvironment{proofT}{\noindent{\bf
Proof }} {\hspace*{\fill}$\Box_{Theorem~\ref{\toto}}$\par\vspace{3mm}}
\newenvironment{proofL}{\noindent{\bf
Proof }} {\hspace*{\fill}$\Box_{Lemma~\ref{\toto}}$\par\vspace{3mm}}
\newenvironment{proofC}{\noindent{\bf
Proof }} {\hspace*{\fill}$\Box_{Corollary~\ref{\toto}}$\par\vspace{3mm}}
\newenvironment{proofP}{\noindent{\bf
Proof }} {\hspace*{\fill}$\Box_{Property~\ref{\toto}}$\par\vspace{3mm}}
\newenvironment{proofCl}{\noindent{\bf
Proof }} {\hspace*{\fill}$\Box_{Claim~\ref{\toto}}$\par\vspace{3mm}}

\newenvironment{lemma-repeat}[1]{\begin{trivlist}
\item[\hspace{\labelsep}{\bf\noindent Lemma~\ref{#1} }]}%
{\end{trivlist}}

\newenvironment{theorem-repeat}[1]{\begin{trivlist}
\item[\hspace{\labelsep}{\bf\noindent Theorem~\ref{#1} }]}%
{\end{trivlist}}

\newenvironment{claim-repeat}[1]{\begin{trivlist}
\item[\hspace{\labelsep}{\bf\noindent Claim~\ref{#1} }]}%
{\end{trivlist}}

\newcounter{linecounter}
\newcommand{\linenumbering}{\ifthenelse{\value{linecounter}<10}{(0\arabic{linecounter})}{(\arabic{linecounter})}}
\renewcommand{\line}[1]{\refstepcounter{linecounter}\label{#1}\linenumbering}
\newcommand{\resetline}[1]{\setcounter{linecounter}{0}#1}
\renewcommand{\thelinecounter}{\ifnum \value{linecounter} > 9\else 0\fi \arabic{linecounter}}

%
\newcommand{\comment}[1]{\textcolor{red}{\textbf{#1}}}

\title{\bf Specifying Concurrent Problems: \\ Beyond Linearizability}


\author{Armando Casta\~{n}eda$^{\dag}$~~
        Sergio Rajsbaum$^{\dag}$~~
        Michel Raynal$^{\star,\circ}$\\~\\
$^{\dag}$   Instituto de Matem\'aticas, UNAM, M\'exico D.F, 04510, M\'exico\\
$^{\star}$  Institut Universitaire de France\\
$^{\circ}$  IRISA, Universit\'e de Rennes 35042 Rennes Cedex, France \\
{\small {\tt armando.castaneda@im.unam.mx~~
           rajsbaum@im.unam.mx~~ 
           raynal@irisa.fr}}
}
\date{}
\maketitle

\begin{abstract}
Tasks and objects are two predominant ways of specifying distributed problems. 
A \emph{task}  is specified by an input/output relation, defining for each 
set of processes that may run concurrently, and each assignment of inputs to 
the processes in the set, the valid outputs of the processes. 
An \emph{object} is specified by an automaton describing the outputs 
the object may produce when it is accessed sequentially. Thus,
tasks explicitly state what may happen only when sets of processes
 run concurrently, while objects only specify what happens when processes 
access the object sequentially. Each one requires its own \emph{implementation} notion, to tell when an execution satisfies the specification.
For objects  \emph{linearizability}  is commonly used,  a very elegant and useful consistency condition.  
For tasks  implementation notions are less explored.

These two orthogonal approaches are central, the former in distributed 
computability, and the later in concurrent programming, yet they have 
not been unified.
Sequential specifications are very convenient, especially important is
the {\it locality} property of linearizability, which states that one can build 
systems in a modular way, considering object implementations in isolation.
However, many important distributed computing problems, including some well-known
tasks,  have no sequential specification. 
Also, tasks are one-shot problems with a semantics that is not fully
understood (as we argue here), and with no clear locality property,
while objects can be invoked in general several times by the same process.

The paper introduces the notion of \emph{interval-sequential} object.
The corresponding implementation notion of  \emph{interval-linearizability}  generalizes
linearizability, and allows to associate states
along the interval of  execution of an operation.
Interval-linearizability allows to specify any task,
however, there are sequential one-shot objects that cannot be expressed as tasks,
under the simplest interpretation of a task. It also shows that
a natural extension of the notion of a task  is expressive enough to specify
any interval-sequential object.

Thus, on the one hand,
interval-sequential linearizability explains in more detail the semantics of
a task, gives a more precise implementation notion, and brings
a locality property to tasks. On the other hand, tasks provide a static
specification for automata-based formalisms.
\end{abstract}
~\\
\noindent
{\bf Keywords}:
asynchronous system, 
concurrent object, 
distributed task,
linearizability, 
object composability, 
sequential specification.

\thispagestyle{empty}
\newpage
\setcounter{page}{1}
\section{Introduction}
\label{sec:introduction}
~

\vspace{-0.7cm}
\paragraph{Concurrent objects and linearizability} 
Distributed computer scientists excel at thinking concurrently, 
and building large distributed programs that work
under difficult conditions with highly asynchronous processes
that may fail. Yet,  they evade thinking
about \emph{concurrent} problem specifications. 
A central paradigm is that of a shared object 
that processes may access concurrently~\cite{HS08,R13,T06}, but the object
is specified in terms of a sequential specification, i.e., 
an automaton describing the outputs the object produces only when it is
accessed sequentially.  Thus, a concurrent algorithm  seeks to emulate an 
allowed sequential behavior. 
 
There are various ways of defining what it means for
an algorithm to \emph{implement} an object, namely, that it satisfies  its
 sequential specification. 
 One of the most popular consistency conditions is 
\emph{linearizability}~\cite{HW90}, 
(see surveys~\cite{DFK13,RS97}).
Given a sequential  specification of an object, 
an algorithm \emph{implements} the object
if every execution can be transformed to a sequential one such that 
(1) it respects the real-time order of invocation and responses and 
(2) the sequential execution is recognized by the automaton specifying 
the object.
It is then said that the corresponding object implementation  is
 \emph{linearizable}.
 Thus, an execution is linearizable if, for each operation call, it is possible to 
 find a unique point in the interval of real-time defined by the invocation and response of the operation,
 and these \emph{linearization points} induce a valid sequential execution.
Linearizability is very popular to design components of large systems
because it is \emph{local}, namely,  one can consider
linearizable object implementations in isolation and \emph{compose} them
without sacrificing linearizability of the whole system~\cite{FVC03}.
Also, linearizability is a \emph{non-blocking} property, which means that
a pending invocation (of a total operation) is never required to wait 
for another  pending invocation to complete.
Textbooks such as~\cite{AW04,HS08,R13,T06} include more detailed discussions
 of linearizability. 

Linearizability has various desirable properties, 
additionally to being local and non-blocking:  it allows talking about the state of an object,
interactions among operations is captured by side-effects on object states; 
documentation size of an object is linear in the number of operations; 
new operations can be added without changing descriptions of  old operations.
 However, as we argue here, linearizability is sometimes too restrictive.
First,  there are problems which have no sequential specifications 
(more on this below).
Second, some problems are more naturally and succinctly defined 
in term of concurrent behaviors.
Third, as is well known, the specification of a problem should be
as general as possible, to allow maximum flexibility to both  programmers and 
program executions.

\paragraph{Distributed tasks}
Another predominant way of specifying a one-shot distributed problem,
 especially in
distributed computability,  is through 
the notion of a \emph{task}~\cite{MW87}.
Several tasks have been intensively studied in distributed computability,
 leading
to an understanding of their relative power~\cite{HRR13}, to
the design of  simulations between models~\cite{BGLR}, and
 to the development of a deep connection between distributed 
computing and topology~\cite{HKRbook}.
Formally, a {task} is specified by an input/output relation, defining for each 
set of processes that may run concurrently, and each assignment of inputs to 
the processes in the set, the valid outputs of the processes. 
Implementation notions for tasks are less explored, and they are not 
as elegant as linearizability.
In practice,   task and implementation are  usually described operationally, somewhat informally.
One of the versions widely used is that an algorithm \emph{implements}
 a task if, 
in every execution where a set of processes participate 
(run to completion, and the other crash from the beginning), 
input and outputs satisfy the task specification.

 A main difference between tasks and objects is how they model the 
concurrency that naturally arises in distributed systems: whiles tasks 
explicitly state what might happen for several (but no all) 
concurrency patterns, objects only specify what happens when 
processes access the object sequentially.

It is remarkable that these two approaches have
largely remained independent\footnote{Also both approaches were proposed 
the same year, 1987, and  both are seminal to their respective research 
areas~\cite{HW87,MW87}.}, 
while the main distributed computing
paradigm, \emph{consensus}, is central to both.
Neiger~\cite{N94} noticed this and proposed a generalization of linearizability called
 \emph{set-linearizability}. He discussed that there are {tasks}, 
like \emph{immediate snapshot}~\cite{BG93},
with no natural specification as sequential objects. 
In this task there is a single
operation ${\sf Immediate\_snapshot}()$, such that
 a snapshot of the shared memory occurs immediately after a write.
If one wants to model immediate snapshot as an object, 
the resulting object implements test-and-set, which is contradictory because
there are read/write algorithms solving the immediate snapshot task and it is well-known that
there are no read/write linearizable implementations of test-and-set.
Thus, it is meaningless to ask if there is a linearizable implementation 
of  immediate snapshot because there is no natural  sequential specification of it.
Therefore, Neiger proposed the notion of a
\emph{set-sequential} object, that allows a set of processes
to access an object simultaneously. Then, one
can define an immediate snapshot set-sequential object,
and there are \emph{set-linearizables} implementations. 

\paragraph{Contributions}
We propose the notion of
an \emph{interval-sequential} concurrent object,
a framework in which an object is specified by an automaton that can
 express any concurrency pattern of 
overlapping invocations of operations, that might occur in an execution 
(although one is not forced to describe all of them). The automaton
is a direct generalization of the automaton of a sequential object, except that
transitions are labeled with sets of invocations and responses,
allowing operations to span several consecutive transitions.
The corresponding implementation notion of  \emph{interval-linearizability}  generalizes
linearizability and set-linearizability, and allows to associate states
along the interval of  execution of an operation.
While  linearizing an execution requires  finding linearization \emph{points}, 
in interval-linearizability one needs to identify a linearization \emph{interval} for each operation (the  intervals might overlap).
Remarkably, this general notion  remains local and non-blocking.
We show that most important tasks  (including \emph{set agreement}~\cite{C93}) have no  specification  neither as a sequential objects nor as a set-sequential objects,
 but they can be naturally expressed as interval-sequential objects.

Establishing the relationship between tasks and (sequential, set-sequential
and interval-sequential) automata-based specifications is subtle,
because tasks admit several natural  interpretations. 
Interval-linearizability is a framework that allows to specify any task,
however, there are sequential one-shot objects that cannot be expressed as tasks,
under the simplest interpretation of a task.
Hence, interval-sequential objects have strictly more power to specify one-shot problems than tasks.
However, a natural extension of the notion of a task  has 
the same expressive power to specify  one-shot concurrent problems, hence strictly
more than sequential and set-sequential objects. 
See Figure~\ref{fig:categories}.
Interval-linearizability goes beyond  unifying sequentially specified objects and tasks, it sheds new light on both of them.
On the one hand,
interval-sequential linearizability provides an explicit operational semantics
to a task (whose semantics, as we argue here, is not well understood), gives a more precise implementation notion, and brings a locality property to tasks. 
On the other hand, tasks provide a static
specification for automata-based formalisms such as sequential, set-sequential
and interval-sequential objects.

\begin{wrapfigure}{r}{-0.3\textwidth}
 \centering
 \vspace{-10mm}
 \includegraphics[width=.44\textwidth]{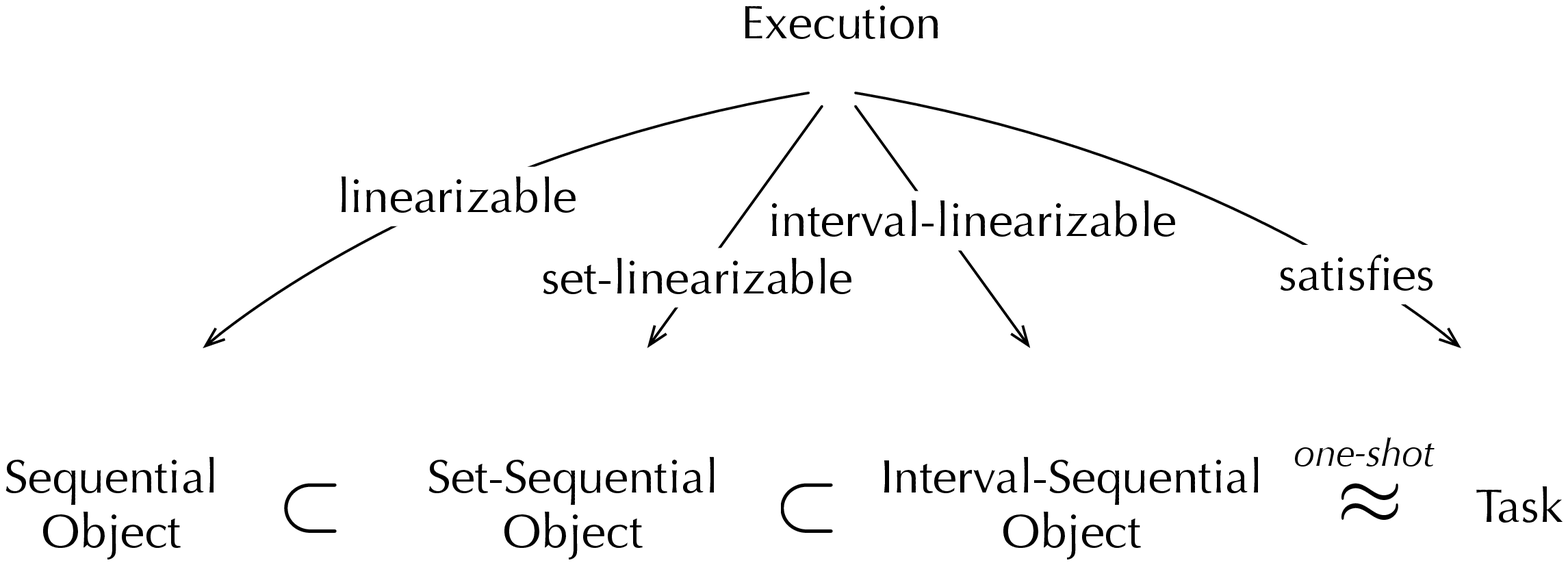}
 \caption{Objects and consistency conditions. 
 The equivalence is between \emph{refined} tasks and one-shot interval-sequential objects.}
 \vspace{-5mm}
 \label{fig:categories}
\end{wrapfigure}

\paragraph{Related work}
Many consistency conditions have been proposed to define the correct 
behavior of sequentially specified objects, that guarantee that all 
the processes see the same sequence  of operations applied to the object. 
Among the most notable are
{\it atomicity}~\cite{L86-a,LH89,M86},  
{\it sequential consistency}~\cite{L79},  
and {\it linearizability}~\cite{HW90}. 
(See surveys~\cite{DFK13,RS97}, and textbooks such 
as~\cite{AW04,HS08,R13,R13-b})\footnote{Weaker consistency 
conditions such as {\it causal consistency}~\cite{ANBHK95}, 
{\it lazy release consistency}~\cite{KCZ92}, or 
{\it eventual consistency}~\cite{W09} are not addressed here.}.   
An extension of linearizability suited to relativistic
distributed systems is presented in~\cite{GG14}. 
 {\it Normality} consistency~\cite{GR99} can be seen
as an extension of linearizability to the  case where an operation
can involve more than one object.  

\Xomit{
{\it Causal consistency}~\cite{ANBHK95} considers objects whose 
correct behavior can be described by a partial order on their operations, 
that respects the basic causal order introduced in~\cite{L78}. 
It is different from ours...

There are also hybrid consistency conditions based on a combination 
of a ``strong'' condition (such as the ones described in the previous item) 
and  causal consistency are investigated in~\cite{F95,FRT15,JFC08}.

Application-aware consistency conditions, which weaken  causal consistency, 
are investigated in~\cite{TBV15}. 
The last category contains consistency conditions ``weaker'' than 
causal consistency. {\it Lazy release consistency}~\cite{KCZ92}, or 
{\it eventual consistency}~\cite{W09} are examples of such conditions. 
[[[should we remove it or else discuss the relation to our work]]]
} 

Neiger proposed unifying sequential objects and tasks, and defined
 set-linearizability~\cite{N94}. In the automaton specifying a set-sequential object, 
 transitions between states  involve more than one operation; these operations are allowed to occur concurrently
and their results can be concurrency-dependent. Thus, linearizability corresponds
to the case when the transitions always involve a single operation.
Later on it was again observered that for some concurrent objects it is impossible to provide a sequential specification, and  similar notion, but based on histories,
was proposed~\cite{HR14} (no properties were proved).
Transforming  the question of wait-free read/write solvability of a one-shot 
sequential object, into the question of solvability of a task was suggested in~\cite{Gafni2014}.
The extension of tasks we propose here is reminiscent to the construction in~\cite{Gafni2014}.

Higher dimensional automata are used to model execution
of concurrent operations, and are the most expressive
model among other common operations~\cite{G2006}. They can
model transitions which consists of sets of operations,
and hence are related to set-linearizability, but do not naturally
model interval-linearizability, and other concerns of concurrent objects.
There is work on  partial order semantics of programs, including
more flexible notions of linearizability, relating  two
arbitrary sets of histories~\cite{FORY2010}. 

\paragraph{Roadmap}
The paper is composed of~\ref{sec:conclusion} sections. 
It considers that the basic definitions related to linearizability are known. 
First, Section~\ref{sec:limit-of-lin} uses a simple example to
illustrate the  limitations of both linearizability and set-linearizability. 
Then, Section~\ref{sec:interval-seq-object} introduces the notion of an 
interval-sequential concurrent object, which makes it possible to 
specify the correct concurrent patterns, without restricting them to be 
sequential patterns. Section~\ref{sec:interval-linearizability} defines 
interval-linearizability and shows it is local and non-blocking. 
Then, Section~\ref{sec:completeness} 
compares the ability of tasks and interval-sequential objects to specify one-shot problems. 
Finally, Section~\ref{sec:conclusion} concludes the paper.

\section{Limitations of linearizability and set-linearizability}
\label{sec:limit-of-lin}
Here we discuss in more detail  limitations of sequential and set-sequential specifications (linearizability  and set-linearizability).
As a running example we use write-snapshot,
a natural task that is implementable from read/write registers
and has no natural specification as a sequential or set-sequential object.
Many other tasks have the same problems. Appendix~\ref{app:addDiscExam} presents other
examples and additional details.

\subsection{The write-snapshot task}
\label{write-snapTask}
\paragraph{Definition and implementation of  write-snapshot}
Sometimes we work with objects with two operations, but that
are intended to be used as one. For instance, a snapshot object~\cite{AADGMS93} 
has  operations ${\sf write}()$ (sometimes called update) and 
${\sf snapshot}()$.
This object has a sequential specification and there are linearizable 
read/write algorithms implementing it (see, e.g.,~\cite{AW04,HS08,R13,T06}).
But many times, a snapshot object is used
in a canonical way, namely,  each time a process invokes  ${\sf write}()$, 
immediately after it  always invokes  ${\sf snapshot}()$. 
Indeed, one would like to think of such an object as providing a single 
operation,  ${\sf write\_snapshot}()$, invoked with a value $x$
to be deposited in the object, and when the operation returns, it gives back
to the invoking process  a snapshot of the contents of the object.
It turns out that this write-snapshot object has neither a  natural 
sequential nor a set-sequential specification. However, it can be specified 
as a task and actually is implementable from read/write registers.

In the \emph{write-snapshot} task, each process $p_i$ starts with a private 
input $v_i$ and outputs a set $set_i$ satisfying the following:

\begin{itemize}
\item Self-inclusion: 
$\langle i,v_i\rangle \in set_i$. 
\item Containment:
$\forall~i,j:~ (set_i\subseteq set_j)\vee  (set_j\subseteq set_i)$. 
\end{itemize}

Note that the specification of write-snapshot is highly concurrent:
it only states what processes might decide when they run until completion, 
regardless of the specific interleaving pattern of invocations and responses.
A simple write-snapshot algorithm based on read/write registers,
 is in Figure~\ref{fig:algorithm-write-snapshot} below.

The \emph{immediate snapshot} task~\cite{BG93} is defined
by adding an Immediacy requirement to the Self-inclusion and Containment
requirements of the   write-snapshot task.
\begin{itemize}
\item
Immediacy:
$\forall~i,j:~ [(\langle j,v_j \rangle \in set_i)\wedge
(\langle i,v_i \rangle \in set_j)] \Rightarrow  (set_i=set_j)$.
\end{itemize}

Figure~\ref{fig:algorithm-write-snapshot} contains an algorithm that
implements write-snapshot (same idea of the well-known algorithm of~\cite{AADGMS93}).
The internal representation of write-snapshot is 
made up of an array of single-writer multi-reader atomic registers 
$\MEM[1..n]$, initialized to $[\bot,\cdots,\bot]$. 
In the following, to simplify the presentation we suppose that the value 
written by $p_i$ is $i$, and the pair $\langle i,v_i\rangle$ 
is consequently denoted $i$. 
When a process $p_i$ invokes  ${\sf write\_snapshot}(i)$, it first 
writes its value $i$ in $\MEM[i]$ (line~\ref{WS-01}).
Then $p_i$ issues repeated  classical ``double collects'' until it obtains 
two successive read of the full array $\MEM$, 
which provide it with the same set of  non-$\bot$ values
(lines~\ref{WS-02}-\ref{WS-05}). When such a successful double collect 
occurs, $p_i$ returns the content of its last read of the array $\MEM$
(line~\ref{WS-06}).  Let us recall that the reading of the $n$ array entries 
are done asynchronously and in an arbitrary order. 
In Appendix~\ref{sec:proof-write-snapshot}, it is shown that this algorithm
implements the write-snapshot task.

\begin{figure}[ht]
\centering{
\fbox{
\begin{minipage}[t]{150mm}
\footnotesize
\renewcommand{\baselinestretch}{2.5}
\resetline
\begin{tabbing}
aaaaa\=aa\=aaa\=aaa\=aaaaa\=aaaaa\=aaaaaaaaaaaaaa\=aaaaa\=\kill 

{\bf operation} ${\sf write\_snapshot}(i)$ {\bf is}   ~~~\% issued by $p_i$ \\

\line{WS-01} \> $\MEM[i] \leftarrow i$;\\

\line{WS-02} \> $new_i \leftarrow 
     \cup_{1\leq j\leq n} \{\MEM[j] \mbox{ such that } \MEM[j]\neq\bot\}$;\\

\line{WS-03} \> {\bf repeat} \= $old_i \leftarrow new_i$;\\

\line{WS-04} \> \>
$new_i \leftarrow  
     \cup_{1\leq j\leq n} \{\MEM[j] \mbox{ such that } \MEM[j]\neq\bot\}$\\

\line{WS-05} \> {\bf until} $(old_i=new_i)$  {\bf end repeat};\\

\line{WS-06} \> ${\sf return}(new_i)$.

\end{tabbing}
\normalsize
\end{minipage}
}
\caption{A write-snapshot algorithm} 
\label{fig:algorithm-write-snapshot}
}
\end{figure}

\paragraph{Can the write-snapshot task be specified as a sequential object?}
Suppose there is a deterministic sequential specification of write-snapshot. 
Since the write-snapshot task is implementable from read/write registers,
one expects that there is a linearizable algorithm $A$
implementing the write-snapshot task from read/write registers.
But $A$ is linearizable, hence any of its executions can be seen
as if all invocations occurred one after the other, in some order.
Thus, always there is a first invocation,
which must output the set containing only its input value.
Clearly, using $A$ as a building block, one can trivially solve test-and-set.
This  contradicts  the fact that test-and-set cannot be implemented 
from read/write registers. The contradiction comes from the fact that,
 in a deterministic  sequential specification of write-snapshot, the values in the output set of a process
can only contain input values of operations that happened before.
Such a specification is actually modelling  
a proper subset of all possible relations between inputs and outputs,
of the distributed problem we wanted to model at first.
This phenomenon is more evident when we consider the execution in
Figure~\ref{fig:simple-counterEx}, which can be produced by the write-snapshot 
algorithm in Figure~\ref{fig:algorithm-write-snapshot} in the Appendix.

\begin{figure*}[th]
\centering{
\scalebox{0.4}{\input{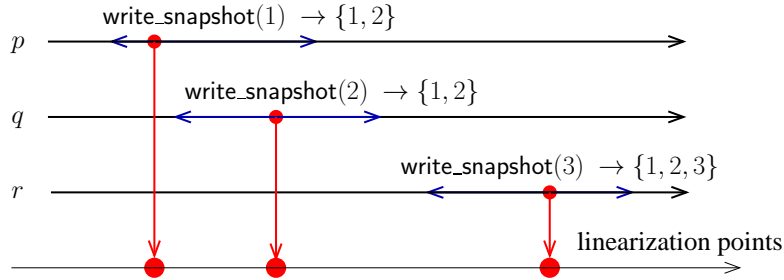}}
\caption{A linearizable write-snapshot execution that predicts the future}
\label{fig:simple-counterEx}
}
\end{figure*}

Consider a non-deterministic sequential specification of write-snapshot (the automaton is in
Appendix~\ref{sec:proof-write-snapshot}).
When linearizing the execution in Figure~\ref{fig:simple-counterEx},
one has to put either the invocation of $p$ or $q$ first,
in either case the resulting sequential execution seems to say that the first process predicted the future 
and knew that $q$ will invoke the task.
The linearization points in the figure describe a possible sequential
ordering of operations.
These anomalous future-predicting sequential specifications result 
in linearizations points without the intended meaning of ``as if the
operation was atomically executed at that point."

\vspace{0.2cm}
\begin{figure*}[th]
\centering{
\scalebox{0.4}{\input{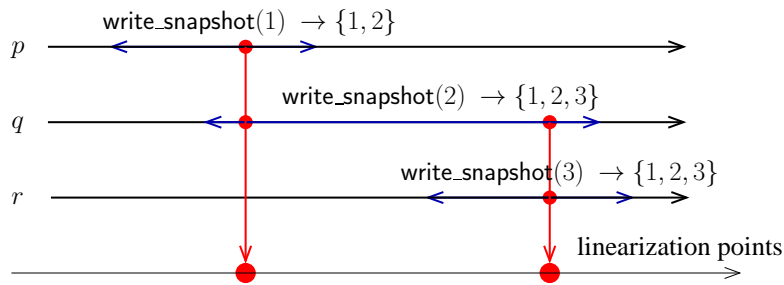}}
\caption{A write-snapshot execution that is not set-linearizable}
\label{fig:simple-counterEx2}
}
\end{figure*}

\paragraph{Why set-linearizability is not enough}
Neiger noted the problems with the execution in 
Figure~\ref{fig:simple-counterEx} discussed above, in the
context of the immediate snapshot task.
He proposed in~\cite{N94} the idea that
a specification should allow to express that sets of operations that
can be concurrent. He called this notion {\it set-linearizability}.
In set-linearizability, an execution accepted by a \emph{set-sequential} 
automaton is  a sequence of non-empty sets with operations, 
and each set denotes operations that are executed concurrently.
In this way, in the execution in Figure~\ref{fig:simple-counterEx}, the
operations of $p$ and $q$ would be set-linearized together, and then 
the operation of $r$ would be set-linearized alone at the end.
While set-linearizability is sufficient to model the immediate-snapshot
task, it is not enough for specifying most other tasks.

Consider the write-snapshot task.
In set-linearizability, in the execution in Figure~\ref{fig:simple-counterEx2} 
(which can be produced by the write-snapshot algorithm, but is not
a legal immediate snapshot execution),
one has to decide if the operation of $q$ goes together with the one 
of $p$ or $r$. In either case,  in the resulting execution 
a process seems to predict a future operation.
In this case the problem is that there are operations that are
affected by several operations that are not concurrent 
(in Figure~\ref{fig:simple-counterEx2},  $q$ is affected by both 
$p$ and $r$, whose operations are not concurrent). 
This cannot be expressed as a set-sequential execution.
Hence, to succinctly express this type of behavior, we need a more 
flexible framework in which it is possible to express that an operation 
happens in an interval of time that can be affected by several operations.

\subsection{Additional examples of tasks with no sequential specification and a potential solution}
\label{sec:addExampNoSeq}
As we shall see, most tasks are problematic for dealing with them
through linearizability,  and have  no deterministic  sequential
specifications. Some have been studied in the past,
such as the following, discussed in more detail in Appendix~\ref{app:addExampNoSeq}.
\begin{itemize}
\vspace{-0.2cm}
\item 
\emph{adopt-commit}~\cite{G98} is a one-shot shared-memory object
useful  to implement round-based protocols for set-agreement and consensus.
Given an input $u$ to the object,  the result is an output of the form $(commit,v)$ or 
$(adopt,v)$, where $commit/adopt$  is a decision  that indicates whether the process should decide value $v$ immediately or adopt it as its preferred value in later rounds of the protocol.
\vspace{-0.2cm}
\item 
\emph{conflict detection}~\cite{AE14} has been shown to be equivalent to the 
{adopt-commit}. Roughly, if at least two different values are proposed concurrently
at least one process outputs true.
\vspace{-0.2cm}
\item 
 \emph{safe-consensus}~\cite{AGL09}, a weakening of consensus, where the agreement condition
of consensus is retained, but the validity condition  becomes: 
if the first process to invoke it returns before any other process invokes it, 
then it outputs its input; otherwise the consensus output can be arbitrary,
 not even the input of any process.
\vspace{-0.2cm}
\item
 \emph{immediate snapshot}~\cite{BG93},  
which plays an important role in distributed 
computability~\cite{AR02,BG93,SZ00}. A process can write a value
to the shared memory using this operation, and gets back a snapshot
of the shared memory, such that the snapshot occurs immediately after the write.
\vspace{-0.2cm}
\item
\emph{$k$-set agreement}~\cite{C93},
where processes agree on at most $k$ of their input values. 
\vspace{-0.2cm}
\item
\emph{Exchanger}~\cite{HR14},
is a Java object that serves as a synchronization point at which 
threads can pair up and atomically swap elements.
\end{itemize}
\paragraph{Splitting an operation in two}
To deal with these problematic tasks, 
one is tempted to separate an operation into two
operations, {\sf set} and {\sf get}. The first communicates  
the input value of a process, while the second produces an output value to a process.
For instance, $k$-set agreement is easily transformed into an object 
with a sequential
specification, simply by accessing it through  {\sf set} to  
deposit a value into the object and {\sf get} that returns one
of the values in the object.
In fact, every  task can be represented as a sequential object by
splitting the operation of the task in two operations (proof in Appendix~\ref{app:splitting}).

Separating an operation into a proposal operation and a returning operation
 has several problems. First, the program is forced to 
produce two operations, and wait for two responses. There is a 
consequent loss of clarity in the code of the program, in addition to a 
loss in performance, incurred by  a two-round trip delay. 
Also, the intended meaning
of linearization points is lost; an operation is now linearized at \emph{two}
linearization points.
Furthermore, the resulting object may  provably \emph{not} be the same;
a phenomenon that has been observed several times in the context of iterated
 models (e.g., in~\cite{CR14,GR10,RRT08}) 
is that the power of the object can be increased, if one is allowed to 
invoke another object in between the two operations.
 Further discussion of this issue is in Appendix~\ref{app:splitting}.
%
%

\section{Concurrent Objects}
\label{sec:interval-seq-object}
This section defines the notion of an \emph{interval-sequential} concurrent 
object,  which allows to specify behaviors of all the valid 
concurrent operation patterns. These objects include as special
cases sequential and set-sequential objects. 
To this end, the section also describes the  underlying computation model.

\subsection{System model}
\label{sec:sysModel}
The {system}  consists of  $n$ asynchronous 
sequential processes,  $P=\{ p_1,\ldots,p_n\}$,
which communicate through a set
of concurrent objects,  $OBS$.
Each consistency condition  specifies the behaviour of
an object differently, for now we only need to define its interface,
which is common to all conditions.
The presentation follows~\cite{CHJT04,HW90,R13}.

Given a set  $OP$ of  operations  offered  by the objects of the system to
the processes $P$, 
let $Inv$ be the set of all invocations to operations that can be issued 
by a process in a system, and $Res$ be the set of all responses to the 
invocations in $Inv$. There are functions 
\begin{equation} \label{eq1}
\begin{split}
id: Inv &  \rightarrow P \\
op: Inv &  \rightarrow OP\\
op: Res &  \rightarrow OP \\
res: Res &  \rightarrow Inv\\
obj: OP &  \rightarrow OBS
\end{split}
\end{equation}
where $id(in)$ tells which process invoked $in\in Inv$,
$op(in)$ tells which operation was invoked, 
$op(r)$ tells which operation was responded,
$res(r)$ tells which invocation corresponds to  $r\in Res$,
and $obj(oper)$ indicates the object that offers operation $oper$ .
There is an induced function $id: Res \rightarrow P$ defined by 
 $id(r) = id(res(r))$. Also,
 induced functions 
 $obj: Inv \rightarrow OBS$ defined by  $obj(in) = obj(op(in))$,
and
 $obj: Res \rightarrow OBS$ defined by  $obj(r) = obj(op(r))$.
The set of operations of an object $X$, $OP(X)$, consists of all operations
$oper$, with $obj(oper)=X$. Similarly, $Inv(X)$ and $Res(X)$ are
resp. the set of invocations and responses of $X$.

A \emph{process} is a deterministic automaton that interacts with 
the objects in $OBS$.
It  produces a sequence of steps, 
where a \emph{step} is an invocation of an object's operation, or reacting to
an object's response (including  local processing). 
Consider the set of all operations $OP$ of objects in $OBS$,
and all the corresponding possible invocations $Inv$ and responses $Res$.
A process $p$ is an automaton $(\Sigma,\nu,\tau)$,  with states $\Sigma$ and
functions $\nu,\tau$ that describe the interaction of the process
with the objects. Often there is also a set of initial states 
$\Sigma_0\subseteq\Sigma$.
Intuitively, if $p$ is in state $\sigma$ and $\nu(\sigma)=(op,X)$
then in its next step $p$ will apply operation $op$ to object $X$.
Based on its current state, $X$ will return a response $r$ to $p$ and will
enter a new state, in accordance to its transition relation. Finally,
$p$ will enter state $\tau(\sigma,r)$ as a result of the response it
received from $X$.
\begin{wrapfigure}{r}{0.3\textwidth}
 \centering
 \vspace{-2mm}
 \includegraphics[width=.43\textwidth]{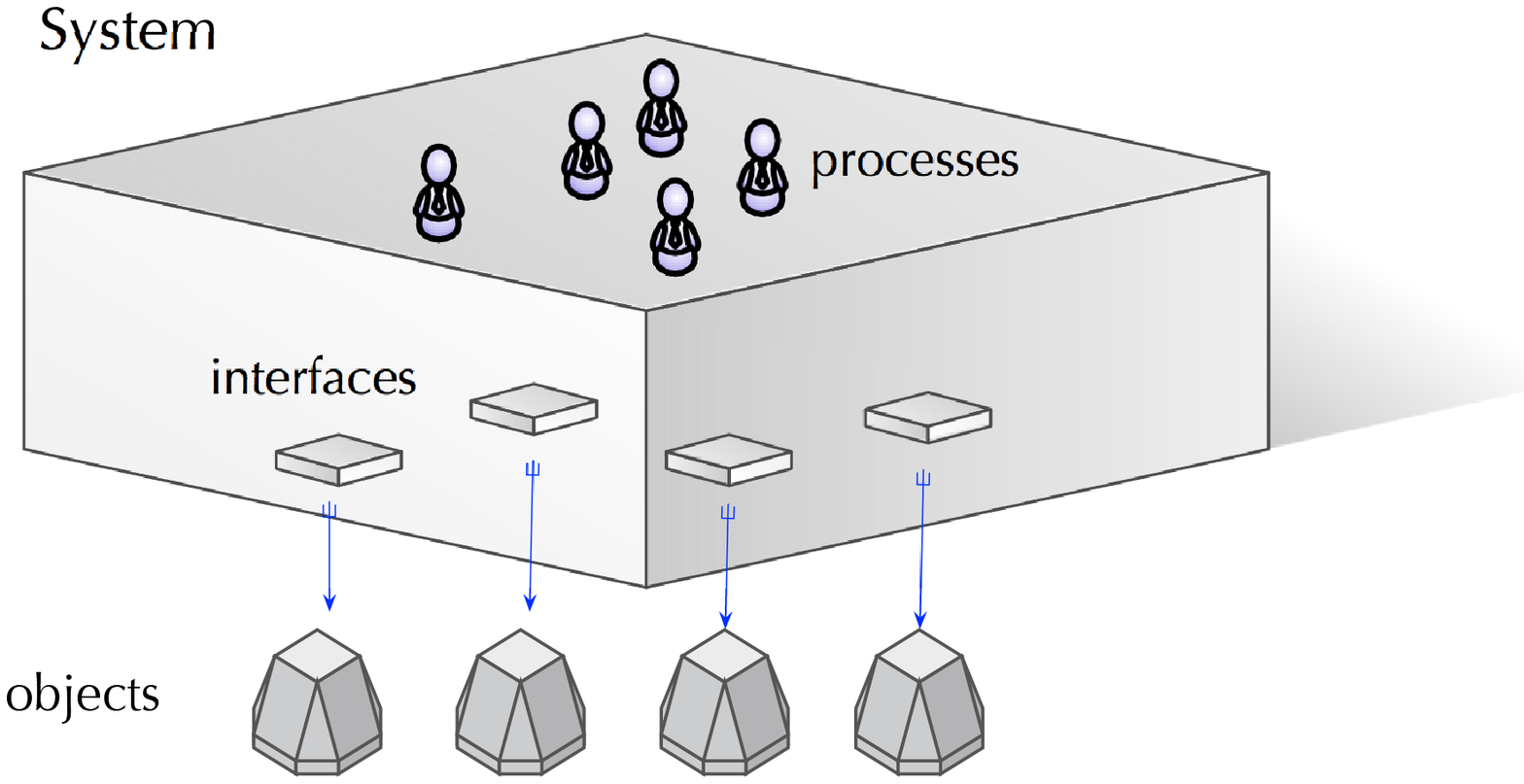}
 \vspace{-5mm}
 \label{fig:system}
\end{wrapfigure}

Finally, a \emph{system} consists of a set of processes, $P$, a set 
of objects $OBS$ so that each $p\in P$ uses a subset of $OBS$,
together with an initial state for each of the objects.

A \emph{configuration} is a tuple consisting of the state of each process
and each object, and a configuration is \emph{initial}
if each process and each object is in an initial state.
An \emph{execution} of the system is modelled by a sequence of events $H$ 
arranged in a total order $\widehat{H}=(H,<_H~)$,
where each event is an invocation $in\in Inv$ or a 
response  $r\in Res$, that can be produced following the
process automata, interacting with the objects.
Namely, an execution starts, given any initial configuration,
by having any process  invoke an operation, according to its transition relation.
In general, once a configuration is reached, the next event
can be a response from an object to an operation of a process
or an invocation of an operation by a process whose last invocation
has been responded.
Thus, an execution is well-formed, in the sense that it consists
of an interleaving of invocations and responses to operations,
where a processes invokes an operation only when its last invocation
has been responded.


\subsection{The notion of an {\it Interval-sequential} object}
\label{sec:intSeqObj}
To generalize  the usual notion
of a sequential object e.g.~\cite{CHJT04,HW90} 
(recalled in Appendix~\ref{sec:linearizability}), instead
of considering sequences of invocations and responses, we consider sequences of
\emph{sets} of invocations and responses. 
An \emph{invoking concurrency class}  $C\subseteq 2^{Inv}$,
is a non-empty subset of $Inv$ 
such that  $C$ contains at most one invocation by the
same process.
A \emph{responding concurrency class} $C$,    $C\subseteq 2^{Res}$, is defined similarly.

\paragraph{Interval-sequential execution}
An \emph{interval-sequential execution} $h$ is an alternating sequence
of invoking and responding concurrency classes, starting in an invoking class,
 $h=I_0,R_0,I_1,R_1,\ldots,I_m,R_m$,
where  the following conditions are satisfied
\begin{enumerate}
\item 
For each $I_i\in h$, any two invocations $in_1,in_2\in I_i$ are by different
processes, $id(in_1)\neq id(in_2)$. Similarly,  for $R_i\in h$ if 
$r_1,r_2\in R_i$ then $id(r_1)\neq id(r_2)$, 
\item 
Let $r\in R_i$ for some $R_i\in h$.
Then there is $in\in I_j$ for some $j\leq i$, such that $res(r)=in$
and furthermore, there is no other $in'$ with $id(in)=id(in')$ with
$in'\in I_{j'}$, $j<j'\leq i$.
\end{enumerate}
It follows that  an execution $h$ consists of matching
invocations and responses,  perhaps with some pending invocations
with no response.

\paragraph{Interval-sequential object}
An \emph{interval-sequential} object $X$ is a (not necessarily finite)
Mealy state machine $(Q,2^{Inv(X)},2^{Res(X)},\delta)$  whose output values $R$
are responding concurrency classes $R$ of $X$, $R\subseteq  2^{Res(X)}$,
are determined both by its current state $s\in Q$ and the current
input $I\in 2^{Inv(X)}$, where $I$ is an invoking concurrency class of $X$.
There is a set of  \emph{initial states} $Q_0$ of $X$,  
$Q_0\subseteq Q$.
The transition relation $\delta\subseteq Q\times 2^{inv(X)}\times 2^{Res(X)}\times Q$ 
specifies both, the output of the automaton and its next state.
If $X$ is in state $q$ and it receives as input a set of invocations $I$,
then, if  $(R,q') \in \delta(q,I)$, the meaning is that $X$ may return
the non-empty set of responses $R$ and move to state $q'$. 
We stress that always both $I$ and $R$ are non-empty sets.

 \paragraph{Interval-sequential execution of an object}
Consider an initial state $q_0\in Q_0$ of $X$ and a
sequence of inputs  $I_0,I_1,\ldots I_m$.
Then a sequence of outputs that $X$ may produce is $R_0,R_1,\ldots R_m$,
where $(R_i,q_{i+1})\in \delta(q_i,I_i)$.
Then the \emph{interval-sequential execution of $X$} starting in $q_0$ is
 $q_0,I_0,R_0,q_1,I_1,R_1,\ldots,q_m,I_m,R_m$. However, we
require  that the object's response at a state
uniquely determines the new state, i.e. we assume if $\delta(q,I_i)$ contains 
$(R_i,q_{i+1})$ and $(R_i,q'_{i+1})$ then $q_{i+1}=q'_{i+1}$.
 Then we may denote the interval-sequential execution of $X$,  starting in $q_0$ by
 $h=I_0,R_0,I_1,R_1,\ldots,I_m,R_m$, because the 
 sequence of states $q_0,q_1,\ldots,q_m$ is
uniquely determined by $q_0$, and by the sequences of inputs and responses.
When we omit mentioning $q_0$ we assume there is some initial state in $Q_0$
that can produce $h$.

Notice that $X$ may be non-deterministic, in a given state $q_i$ with 
input $I_i$
it may move to more than one state and return more than one response.
Also, sometimes it is convenient to require that the object is \emph{total}, meaning that,
for every singleton set $I \in 2^{Inv}$ and every state $q$
in which the invocation $inv$ in $I$ is not pending,
there is an $(R,q') \in  \delta(q,I)$ in which there is a response to $inv$ in $R$.

Our definition of  interval-sequential execution is motivated by the fact
that we are interested in \emph{well-formed} executions
$h=I_0,R_0,I_1,R_1,\ldots,I_m,R_m$. Informally, the processes
should behave well, in the sense that a process does not invoke
a new operation before its last invocation received a response.
Also, the object should behave well, in the sense that it should not
 return a response to an operation that  is not  pending.

The \emph{interval-sequential specification} of $X$,  $ISSpec(X)$, 
is the set of all its interval-sequential executions.


\paragraph{Representation of interval-sequential executions}
In general, we will be thinking of an interval-sequential execution $h$ 
as an alternating sequence
of invoking and responding concurrency classes starting with an invoking class,
$h=I_0,R_0,I_1,R_1,\ldots,I_m,R_m$.
However, it is sometimes convenient to think of an execution as a 
a total order $\widehat S = (S, \stackrel{S}{\longrightarrow} )$ on a
subset $S \subseteq CC(X)$, where $CC(X)$, 
is the set with all invoking and responding concurrency classes of $X$; namely,
$h=I_0  \stackrel{S}{\longrightarrow} R_0 \stackrel{S}{\longrightarrow}I_1 \stackrel{S}{\longrightarrow}R_1 \stackrel{S}{\longrightarrow}\cdots \stackrel{S}{\longrightarrow} I_m \stackrel{S}{\longrightarrow}R_m$.

 In addition, the execution $h=I_0,R_0,I_1,R_1,\ldots,I_m,R_m$
can be represented by a table, with a column for each element in
the sequence $h$, and a row for each process. A member  $in\in I_j$ 
invoked by $p_k$ (resp. a response $r\in R_j$ to $p_k$) 
is placed in the $k$'th row, at the $2j$-th column (resp. $2j+1$-th column).  
Thus, a transition of the automaton will correspond to two consecutive
columns, $I_j,R_j$. See Figure~\ref{fig:fig-validity-wAutomata},
and several more examples  in the figures below.

Interval-sequential objects include as particular cases
set-sequential and sequential objects, as illustrated in Figure~\ref{fig:categories}.

\begin{remark}[Sequential  and Set-sequential objects]
\label{remark:seqAndSetSeq}
Let  $X$ be an interval-sequential object, $(Q,2^{Inv(X)},2^{Res(X)},\delta)$.
Suppose for all states $q$ and all $I$,
if  $\delta(q,I)=(R,q')$, then $|R|=|I|$, and additionally 
each $r\in R$ is a response to one $in\in I$. Then $X$ is a \emph{set-sequential} object.
If in addition,  $|I|=|R|=1$,
 then $X$ is a sequential object in the usual sense.
\end{remark}

\subsection{Examples: Validity and validity with abort}
\label{sec:valObj}
 Consider an object $X$ with a single operation ${\sf validity}(x)$,
that can be invoked by each process, with a \emph{proposed} input parameter $x$,
and a very simple specification: an operation returns a value that 
has been proposed. This problem is easily specified as a task,
see Appendix~\ref{app:validity}.
Indeed, many tasks include  
this property, such as consensus, set-agreement, etc.
As an interval-sequential object, it is formally specified 
by an automaton,
where each state $q$ is labeled with two values, $q.vals$ is the set
of values that have been proposed so far, and $q.pend$ is the set of
processes with pending invocations. The initial state $q_0$
has $q_0.vals=\emptyset$ and $q_0.pend=\emptyset$.
If $in$ is an invocation to the object, let $val(in)$ be the proposed value,
and if $r$ is a response from the object, let $val(r)$ be the responded value.
For a set of invocations $I$ (resp. responses $R$) $vals(I)$ denotes
the proposed values in $I$ (resp. $vals(R)$).
The transition relation  $\delta(q,I)$ contains all pairs $(R,q')$ such that:
\begin{itemize}
\item 
If $r\in R$ then $id(r)\in q.pend$ or there is an $in\in I$ with $id(in)=id(r)$,
\item If $r\in R$ then $val(r)\in q.vals$ or there is an $in\in I$ with 
$val(in)=val(r)$, and
\item 
$q'.vals= q.val\cup vals(I)$ and
 $q'.pend= (q.pend \cup ids(I))\setminus ids(R)$.
\end{itemize}
On the right of Figure~\ref{fig:fig-validity-wAutomata}
there is part of a validity object automaton.
On the left of Figure~\ref{fig:fig-validity-wAutomata} is illustrated
an interval-sequential execution with the vertical red double-dot lines:
$I_0,R_0,I_1,R_1$, where 
$I_0=\{p.{\sf validity}(1),q.{\sf validity}(2)\}$,
$R_0=\{p.{\sf resp}(2)\}$,
$I_1=\{r.{\sf validity}(3)\}$,
$R_1=\{q.{sf resp}(3),r.{\sf resp}(1)\}$.

 The interval-linearizability consistency notion described in Section~\ref{sec:interval-linearizability} will formally define how a general execution (blue double-arrows in the figure) can
 be represented by an interval-sequential  execution (red double-dot lines), and hence tell if 
it satisfies the validity object specification.
Notice that the execution in Figure~\ref{fig:fig-validity-wAutomata} shows that the validity object has no specification 
 neither as a sequential nor as a set-sequential object, for 
 reasons similar to those  discussed in Section~\ref{write-snapTask} about
 Figure~\ref{fig:simple-counterEx2}.

\begin{figure*}
\begin{center}
	\epsfxsize=6.0in
	\epsfbox{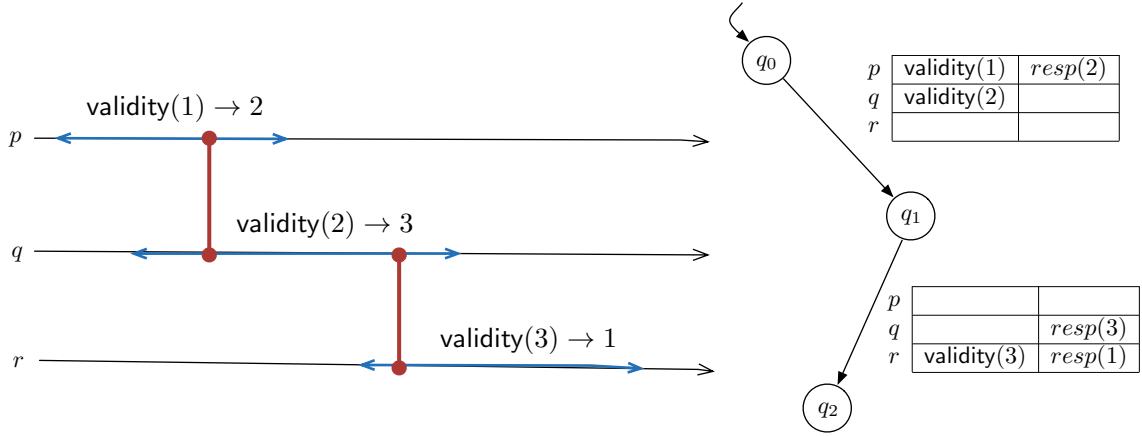}%
    \end{center}
\caption{An execution of a validity object, and the corresponding
part of an interval-sequential automata} 
\label{fig:fig-validity-wAutomata}
\end{figure*}


\paragraph{Augmenting the validity object with an ${\sf abort}()$ operation}
As an illustration of the expressiveness of an interval-sequential automaton,
let us add an operation denoted ${\sf abort}()$  to the validity object,
to design a \emph{validity $k$-abort object.}
Since the validity object is not set-linearizable, neither is  the validity
with abort object.
Intuitively, a process can invoke ${\sf abort}()$ to ``block" the object,
but this might happen only if there are at least $k$ concurrent abort operations.
The operation  ${\sf abort}()$ 
returns either $aborted$ or $notAborted$, to indicate its result.
If all the concurrent ${\sf abort}()$  operations return $aborted$, then
any operation happening together or after them,
returns $aborted$ as well.
Hence, if only one process invokes ${\sf abort}()$  then
the object behaves as a $Validity$ object. 
How do we formally argue that the execution in Figure~\ref{fig:fig-validityAbort} is correct? Interval-Linearizability is 
a correctness implementation notion  that serves this purpose, defined next.
In Appendix~\ref{app:safe-cons-and-abort}, the validity object is formally defined.

\begin{figure*}
\begin{center}
	\epsfxsize=4.0in
	\epsfbox{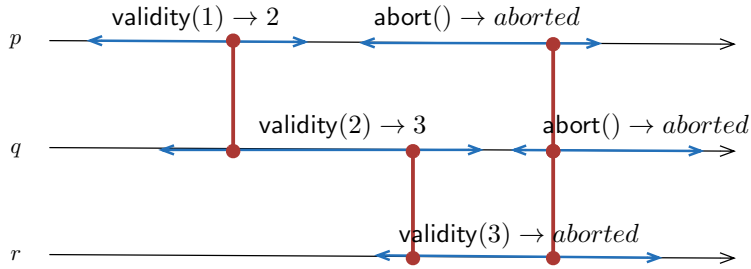}%
    \end{center}
\caption{An execution of a Validity-Abort object (1)}
\label{fig:fig-validityAbort}
\end{figure*}


\section{Interval-Linearizability}
\label{sec:interval-linearizability}

We first define interval-linearizability and then
prove it is  local  and non-blocking.

\subsection{The notion of interval-linearizability}

\paragraph{Interval-sequential execution of the system}
Consider a subset $S \subseteq CC$ of  the concurrency classes of the objects $OBS$ in the system and an interval-sequential execution
 $\widehat S = (S, \stackrel{S}{\longrightarrow})$,
defining an alternating sequence of invoking and responding
concurrency classes, starting with an invoking class.
For an object $X$, the \emph{projection of $\widehat S$ at $X$}, 
$\widehat S |_X = (S_X, \stackrel{S_X}{\longrightarrow} )$, 
is defined as follows:
(1) for every $C \in S$ with at least one invocation or response on $X$, $S_X$ contains a concurrency class $C'$, consisting of the  (non-empty) subset
of $C$ of all invocations or responses of $X$, and
(2) for every $C', C'' \in S_X$, $C' \stackrel{S_X}{\longrightarrow} C''$
if and only if there are $T', T'' \in S$ such that 
$C' \subseteq T'$, $C'' \subseteq T''$ and 
$T' \stackrel{S}{\longrightarrow} T''$.

We say that $\widehat S = (S, \stackrel{S}{\longrightarrow} )$
is an \emph{interval-sequential execution of the system}
if $\widehat S |_X$ is an interval-sequential execution of $X$ for every  $X\in OBS$.
That is, if $\widehat S |_X\in ISSpec(X)$,
the {interval-sequential specification} of $X$, for every $X\in OBS$.
Let $\widehat S = (S, \stackrel{S}{\longrightarrow})$ be an 
interval-sequential execution.
For a process $p$, the \emph{projection of $\widehat S$ at $p$},
$\widehat S |_{p} = (S_{p}, \stackrel{S_{p}}{\longrightarrow})$,
is defined as follows:
(1) for every $C \in S$ with an invocation or response by $p$, $S_{p}$ contains 
contains a class $C$ with the  invocation or 
response by $p$
(there is at most one event by $p$ in $C$), and
(2) for every $a, b \in S_{p}$, $a \stackrel{S_{p}}{\longrightarrow} b$
if and only if there are $T', T'' \in S$ such that 
$a \in T'$, $b \in T''$ and $T' \stackrel{S}{\longrightarrow} T''$.


\paragraph{Interval-linearizability} 
Recall  that an execution of the system is a sequence of 
invocations and responses (Section~\ref{sec:sysModel}). 
An invocation in an execution $E$ is \emph{pending} if it 
has no matching response, otherwise it is \emph{complete}. 
An \emph{extension} of an execution $E$ is obtained by appending 
zero or more responses to pending invocations.

An \emph{operation call} in $E$ is a pair consisting of an 
invocation and  its matching response.
 Let $comp(E)$ be the sequence obtained 
from $E$ by removing its pending invocations. 
The order in which invocation and responses 
in $E$  happened,
induces the following partial order: 
$\widehat{OP} = (OP, \stackrel{op}{\longrightarrow})$
where $OP$ is the set with all operation calls in $E$, and 
for each pair ${\sf op}_1, {\sf op}_2 \in OP$,
 ${\sf op}_1 \stackrel{op}{\longrightarrow} {\sf op}_2$ 
if and only if $term({\sf op}_1) < init({\sf op}_2)$ in $E$,
namely, the response of $op_1$ appears before the invocation
of $op_2$. Given two operation $op_1$ and $op_2$, 
$op_1$ \emph{precedes} $op_2$ 
if ${\sf op}_1 \stackrel{op}{\longrightarrow} {\sf op}_2$,
and they are \emph{concurrent} 
if ${\sf op}_1 \stackrel{op}{\nrightarrow} {\sf op}_2$
and ${\sf op}_2 \stackrel{op}{\nrightarrow} {\sf op}_1$.

Consider an execution of the system $E$ and its associated partial order 
$\widehat{OP} = (OP, \stackrel{op}{\longrightarrow})$,
and let $\widehat S = (S, \stackrel{S}{\longrightarrow} )$ be 
an interval-sequential execution.
We say that an operation $a \in OP$ \emph{appears} in a concurrency 
class $S' \in S$
if its invocation or response is in $S'$. Abusing  notation, 
we write $a \in S'$.
We say that $\stackrel{S}{\longrightarrow}$ 
\emph{respects} $\stackrel{op}{\longrightarrow}$,
also written as $\stackrel{op}{\longrightarrow} \subseteq 
\stackrel{S}{\longrightarrow}$,
if for every $a,b \in OP$ such that $a \stackrel{op}{\longrightarrow} b$,
for every $T', T'' \in S$ with $a \in T'$ and $b \in T''$,
it holds that $T' \stackrel{S}{\longrightarrow} T''$.

\begin{definition}[Interval-linearizability]
An execution $E$ is \emph{interval-linearizable} if there is an extension 
$\overline E$
of $E$ and an interval-sequential execution 
$\widehat{S} = (S, \stackrel{S}{\longrightarrow} )$ such that
\begin{enumerate}
\item for every process $p$, $comp(\overline E)|_p = \widehat{S}|_p$,
\item for every object $X$, $\widehat{S}|_X \in ISS(X)$ and
\item $\stackrel{S}{\longrightarrow}$ respects 
$\stackrel{op}{\longrightarrow}$, where
$\widehat{OP} = (OP, \stackrel{op}{\longrightarrow})$ 
is the partial order associated to $comp(\overline E)$.
\end{enumerate}
We say that $\widehat{S} = (S, \stackrel{S}{\longrightarrow} )$ is 
an interval-linearization of $E$.
\end{definition}

\begin{remark}[Linearizability and set-linearizability]
\label{remark:seqAndSetLin}
When we restrict to interval-sequential executions in which 
for every invocation there is a response to it in the very next 
concurrency class,
then interval-linearizability boils down to set-linearizability.
If in addition we demand that every concurrency class contains
only one element, then we have linearizability.
See Figure~\ref{fig:categories}.
\end{remark}

We can now complete the example of the validity object.
In Figure~\ref{interval-executionValidity}
there is an interval linearization of the execution
in Figure~\ref{fig:fig-validity-wAutomata}.
Similarly, for the validity with abort object,
in Figure~\ref{interval-executionValidityAbort}
there is an interval linearization of the execution
in Figure~\ref{fig:fig-validityAbort}.

\begin{figure*}[th]
\centering{
\hspace{0.0cm}
$\begin{array}{c|c|c|c|c|}
 \cline{2-5}
   & {init} & {term} & {init} & {term}   \\
 \cline{2-5}
 p &  {\sf validity}(1)	&   resp(2)	&  {\,}	& {\,}		 \\
 \cline{2-5}
q &  {\sf validity}(2)  		&   {\,}	&  & resp(3)    \\
 \cline{2-5}
r &   {\,} 		&   {\,}	&  {\sf validity}(3) & resp(1)	 \\
 \cline{2-5}
\end{array}
$
\caption{An execution of a  Validity  object}
\label{interval-executionValidity}
}
\end{figure*}

\begin{figure*}[th]
\centering{
\hspace{0.0cm}
$\begin{array}{c|c|c|c|c|c|c|}
 \cline{2-7}
   & {init} & {term} & {init} & {term} & {init} & {term}  \\
 \cline{2-7}
 p &  {\sf validity}(1)	&   resp(2)	&  {\,}	& {\,}		&  {\sf  abort}()	& resp(aborted) \\
 \cline{2-7}
q &   {\sf validity}(2)  		&   {\,}	&  & resp(3) & {\,}	 {\sf abort}()	& resp(aborted)   \\
 \cline{2-7}
r &   {\,} 		&   {\,}	&  {\sf  validity}(3) & 	 & {\,}		& resp(aborted) 	  \\
 \cline{2-7}
\end{array}
$
\caption{An execution of a  Validity-Abort  object (2)}
\label{interval-executionValidityAbort}
}
\end{figure*}

\subsection{An interval-sequential implementation}

Once we have formally defined the notion of interval-linearizability,
we can show that the write-snapshot algorithm in Section~\ref{write-snapTask}
is interval-linearizable. 

\paragraph{The write-snapshot interval-sequential object}
Here is a formal definition of this task, using an interval-sequential object
based on the validity object of Section~\ref{sec:valObj}.
The  write-snapshot object $X$ has a single operation ${\sf write\_snapshot}(x)$
that can be invoked by each process, with a \emph{proposed} input parameter $x$,
and  returns a set. In the interval-sequential automata
each state $q$ is labeled with two values, $q.vals$ is the set
of id-values that have been proposed so far, and $q.pend$ is the set of
processes with pending invocations. The initial state $q_0$
has $q_0.vals=\emptyset$ and $q_0.pend=\emptyset$.
If $in$ is an invocation to the object, let $val(in)$ be the proposed value,
and $(id(in),val(in)$ be the \emph{proposed id-value pair.}
If $r$ is a response from the object, let $val(r)$ be the responded id-value pair.
For a set of invocations $I$ (resp. responses $R$) $vals(I)$ denotes
the proposed id-value pairs in $I$ (resp. $vals(R)$).
The transition relation  $\delta(q,I)$ contains all pairs $(R,q')$ such that:
\begin{itemize}
\item 
If $r\in R$ then $id(r)\in q.pend$ or there is an $in\in I$ with $id(in)=id(r)$,
\item If $r\in R$ then $val(r)= q.val\cup vals(I)$ 
\item 
$q'.vals= q.val\cup vals(I)$ and
 $q'.pend= (q.pend \cup ids(I))\setminus ids(R)$.
\end{itemize}

An example of an execution an the transitions through the automata
is in Figure~\ref{fig:fig-snapshot-exec}.

\begin{figure*}
\begin{center}
	\epsfxsize=5.50in
	\epsfbox{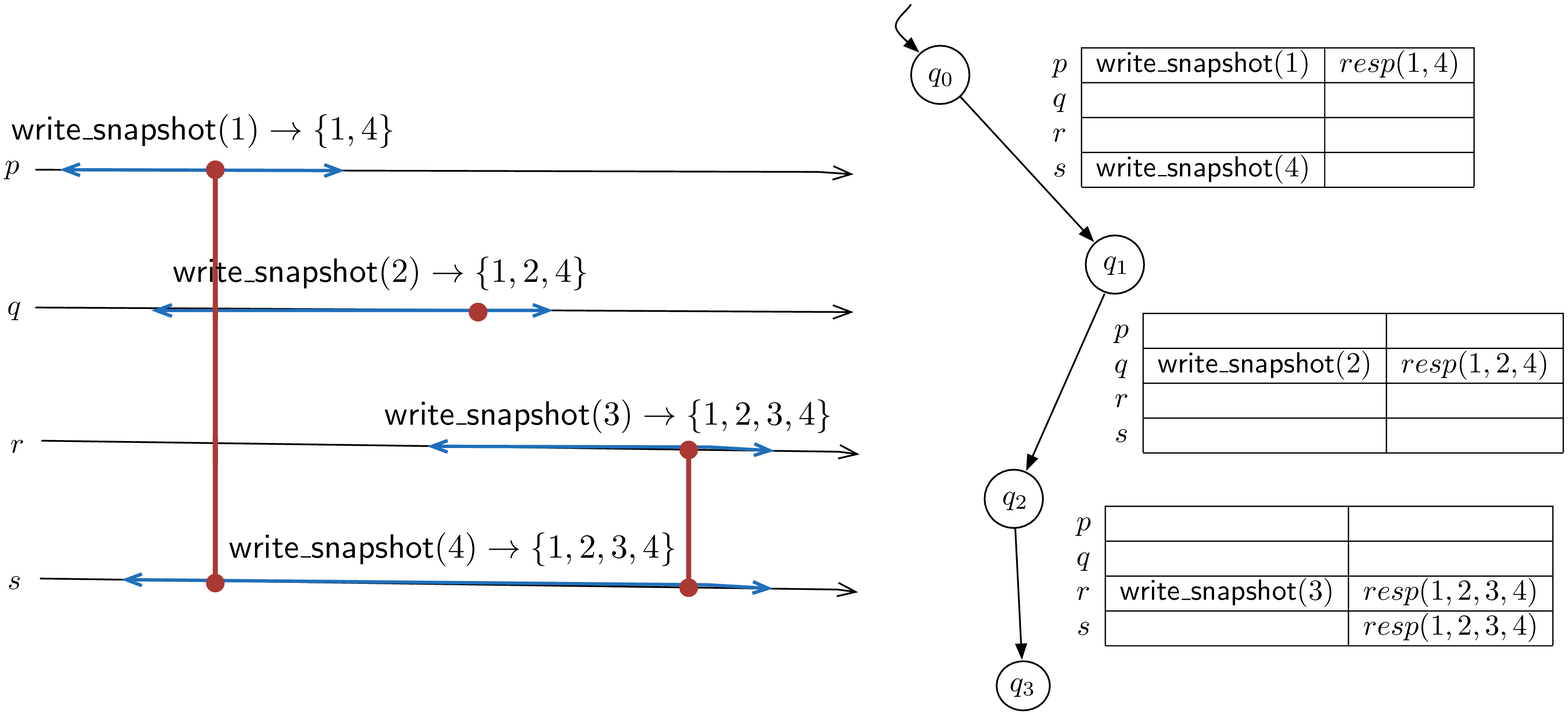}%
    \end{center}
\caption{An execution  of the write-snapshot  task.}
\label{fig:fig-snapshot-exec}
\end{figure*}

\paragraph{The write-snapshot algorithm is interval-linearizable}
The specification of a write-snapshot object
contains every interval-sequential execution satisfying the self-containment and containment properties
(Appendix~\ref{sec:proof-write-snapshot} contains a correctness proof
in the usual style, without interval-linearizability), 
thus, to show that an execution
of the algorithm is interval-linearizable, we need to transform it into a interval-sequential
execution that satisfy the real-time order of invocations and responses.

As with linearizability, interval-linearizability specifies a safety property,
it is not about liveness.
Thus, before showing that
the  algorithm of Figure~{\em\ref{fig:algorithm-write-snapshot}}
is interval-linearizable, we recall the usual 
termination arguments for this style of snapshot algorithm.
The invocation of  ${\sf write\_snapshot}()$  by any 
process $p_i$ terminates, because,
as the number of processes is fixed (equal to $n$), and a 
process invokes  ${\sf write\_snapshot}()$ at most once, it follows that a 
process can  execute at  most $(n-1)$ double collects where each time it
sees new values.

\begin{theorem}
\label{theo:write-snapshot-lin}
The write-snapshot algorithm of Figure~{\em\ref{fig:algorithm-write-snapshot}} is
 interval-linearizable.
\end{theorem}

\begin{proofT}
The proof is very similar to the usual linearizability proof 
for the  obstruction-free implementation of a snapshot object
(we follow~\cite{R13} (Sect.~8.2.1)), except that now \emph{two} points have
to be identified, one for the invocation of an operation and one for the response.

Consider any execution $E$ and let $p_i$ be any process that terminates. As it returns a value $set_i$ 
(line~\ref{WS-06}), we have 
$set_i=old_i=new_i$ where $new_i$ corresponds to the last
asynchronous read of $\MEM[1..n]$ by $p_i$,  and $old_i$ corresponds to
the previous asynchronous read of $\MEM[1..n]$. 
Let $\tau[old_i]$ the time at which terminates the read  of $\MEM[1..n]$
returning $old_i$,  and  $\tau[new_i]$ the time at which starts the read  of 
$\MEM[1..n]$ returning $new_i$.
As $old_i=new_i$, it follows that there is a time $\tau_i$, 
such that  $\tau[old_i] \leq \tau_i \leq \tau[new_i]$ and, due to the 
termination predicate of line~\ref{WS-05}, the set of 
non-$\bot$ values of $\MEM[1..n]$ at time   $\tau_i$ is equal to $set_i$. 

For any process $p_i$ that terminates with $set_i$, we pick a time $\tau_i$ as described above.
Let $\bar{\tau}=\tau_{x_1}\leq \tau_{x_2}\leq \cdots\leq \tau_{x_{m}}$ be the ordered
sequence of chosen times, assuming the number of processes that terminate is $m$
($m\leq n$).
Clearly if $\tau_i = \tau_j$, then $set_i = set_j$, but it is possible
that $set_i = set_j$, with $\tau_i < \tau_j$, in case there is no write
in between $\tau_i$ and  $\tau_j$. Thus, for each  longest subsequence
of times in $\bar{\tau}$ with the same set $set_i$, we pick as representative,
the first time in the subsequence, and consider the following subsequence $\bar{\tau}'$
of $\bar{\tau}$, 
where $p$ ($1\leq p\leq m$) is the number of different sets returned by the processes.
The subsequence is 
$\bar{\tau}'=\tau_{x'_1}< \tau_{x'_2}< \cdots< \tau_{x'_p}$,
where the sets $set_{x'_1},set_{x'_2},\ldots,set_{x'_p}$ are all different.

For each subindex $x'_i$ in $\bar{\tau}'$,  consider the set that is output $set_{x'_i}$.
Let $A_{x'_i}$ be the set of processes in the execution that output  $set_{x'_i}$.
Using these sets and the sequence of times above, we define an interval-sequential execution as follows.
The interval-sequential execution
 $\widehat S = (S, \stackrel{S}{\longrightarrow})$ consists 
of an alternating sequence of invoking and responding
concurrency classes. The first  invoking concurrency class $I_1$ has  all invocations of processes in $set_{x'_1}$,
then $R_1$, the responding concurrency class with all responses by processes in $A_{x'_1}$,
followed by $I_2$, the con currency class with all invocations in $set_{x'_2} \setminus set_{x'_1}$, and 
the responding class with all responses by processes in $A_{x'_2}$, and so on.
For an example, see 
the interval sequential execution in the right of Figure~\ref{fig:fig-snapshot-exec} in Appendix~\ref{sec:proof-write-snapshot}.

If there are pending invocation in $\widehat S$ we just add a responding class in which there is a response
to each of them and they output all values written in the execution.
Observe that $\widehat S$ respects the real-time order of the invocations and responses of $E$
because if the response of $p_i$ precedes the invocation of $p_j$ then
$set_i$ cannot contain $p_j$ and then $\tau_i < \tau_j$, which implies that 
the invocation of $p_j$ in $\widehat S$ happens after the invocation of $p_i$.
Thus, the algorithm is interval-linearizable.

\renewcommand{\toto}{theo:write-snapshot-lin}
\end{proofT}

\subsection{Interval-linearizability is composable and non-blocking}
Even though interval-linearizability is much more general than linearizability
it  retains some of its benefits. Proofs are  in Appendix~\ref{app:completeness}.

\begin{theorem}[Locality of interval-linearizability]
An execution E is interval-linearizable 
if and only if $E|_X$ is interval-linearizable,
for every object $X$.
\end{theorem}

\begin{proof}
We prove  that if each $E|_X$ is interval-linearizable for every $X$,
then $E$ is interval-linearizable (the other direction is trivial).
Consider an interval-linearization 
$\widehat S|_X = (S_X, \stackrel{S_X}{\longrightarrow} )$ of $E|_X$ .
Let $R_X$ be the responses appended to $E_X$ to get $\widehat S |_X$
and let $\overline E$ be the extension of $E$ obtained by appending 
the responses in the sets $R_X$ in some order. 
Let $\widehat{OP} = (OP, \stackrel{op}{\longrightarrow})$ be 
the partial order associated to $comp(\overline E)$.

We define the following relation 
$\widehat S = (S, \stackrel{S}{\longrightarrow} )$. 
The set $S$ is the union of all $S_X$, namely, the
union of all concurrency classes in the linearizations of all objects.
The relation $\stackrel{S}{\longrightarrow}$ is defined as follows:
\begin{enumerate}
\item For every object $X$, 
$\stackrel{S_X}{\longrightarrow} \subseteq \stackrel{S}{\longrightarrow}$.
\item For every pair of  distinct objects $X$ and $Y$, 
for every $a \in OP|_X$ and $b \in OP|_Y$ such that 
$a \stackrel{op}{\longrightarrow} b$ and
$a \in S'$ and $b \in S''$,
for a responding class $S' \in S$ and an invoking class $S'' \in S$,
we define $S' \stackrel{S}{\longrightarrow} S''$.
\end{enumerate}


\begin{claim}
\label{claim-acyclic}
The relation $\stackrel{S}{\longrightarrow}$ is acyclic.
\end{claim}

Although
 $\stackrel{S}{\longrightarrow}$ is acyclic,  it might not be transitive. 
Consider the transitive closure $\stackrel{\overline S}{\longrightarrow}$
of $\stackrel{S}{\longrightarrow}$ . One can easily show that 
$\stackrel{\overline S}{\longrightarrow}$
is acyclic, hence it is a partial order over $S$.
It is well-known that a partial order can be extended to a total order.
Let $\widehat{S}^* = (S, \stackrel{S^*}{\longrightarrow})$ 
a total order obtained from $\stackrel{\overline S}{\longrightarrow}$.
It could be that in $\widehat{S}^*$ concurrency classes do not alternate
 between invoking and responding, however, the first concurrency class
 certainly is an invoking one.
To get an interval-sequential execution, we merge consecutive invoking 
classes and 
responding classes in $\widehat{S}^*$ (namely, we take the union of such 
a sequence) 
and adjust $\stackrel{S^*}{\longrightarrow})$
accordingly. Let $\widehat{T}^* = (T, \stackrel{T^*}{\longrightarrow})$ be the 
resulting interval-sequential execution. We claim that $\widehat{T}^*$ is an
interval-sequential linearization of $E$.

By the definition of $\widehat S = (S, \stackrel{S}{\longrightarrow} )$ above,
we have that for every object $X$, $\widehat{T}^*|_X \in ISS(X)$.
From the assumption that each $\widehat S|_X$ respects 
the real time order in $comp(\overline E|_X)$,
and by the definition of $\widehat S$,
it follows that $\stackrel{T^*}{\longrightarrow}$ respects 
the real time order in $comp(\overline E)$,
namely, $\stackrel{T^*}{\longrightarrow}$ respects 
$\stackrel{op}{\longrightarrow}$.
That and the definition of $\widehat S = (S, \stackrel{S}{\longrightarrow} )$ 
also imply that
for every process $p$, $comp(\overline E)|_p = \widehat{T}^*_p$,
This completes the proof of the lemma.
\end{proof}

When we consider the specification $ISS(X)$ of and interval-sequential
 object with total operation $opName$,
for every $S \in ISS(X)$ and every invocation $\{inv(opName)\}$ to $opName$,
the interval-sequential execution 
$S \, \cdot \, \{inv(opName)\} \, \cdot \, S'$ belongs to $ISS(X)$,
for some responding concurrency class containing a matching response 
to $\{inv(opName)\}$.

\begin{theorem}
\label{theo-non-blocking}
Let $E$ be an interval-linearizable execution in which there is a 
pending invocation $inv(op)$ 
of a total operation. Then, there is a response $res(op)$ such 
that $E \cdot res(op)$ is interval-linearizable.
\end{theorem}

\section{Tasks and their relationship with automata-based specifications}
\label{sec:completeness} 
A task is a static way of specifying a \emph{one-shot} concurrent problem, namely,
a problem with one operation that can be invoked once by each process.
Here we study the relationship between this static way of defining a problem, 
 and the automata-based ways of specifying a problem that we have been considering.
Proofs and additional details are in Appendix~\ref{app:completeness}.

Roughly,  a task $(\m{I},\m{O},\Delta)$ consists of a set of input assignments
$\m{I}$, and a set of output assignments $\m{O}$, which
are defined in terms of sets called \emph{simplexes} of the form
$s=\{(\id_1,x_1),\dots,(\id_k,x_k)\}$. A singleton simplex is a \emph{vertex}.
A simplex  $s$ is used to denote the input values,
or output values in an execution, where  $x_i$ 
denotes the value of the process with identity $\id_i$, either an input
value, or an output value. 
Both $\m{I}$ and $\m{O}$ are \emph{complexes}, which means they
are closed under containment.
There is an input/output relation $\Delta$,
specifying for each input simplex $s\in\m{I}$, a subcomplex of $\m{O}$
consisting of a set of output 
simplexes $\Delta(s)\subseteq\m{O}$ that may be produced with input $s$.
If $s,s'$ are two simplexes in $\m{I}$ with $s'\subset s$,
 then $\Delta(s')\subset\Delta(s)$.
Formal definitions   are in Appendix~\ref{app:tasks}.

\paragraph{When does an execution satisfy a task?}
A task is usually specified informally, in the style of Section~\ref{sec:addExampNoSeq}.
E.g., for the $k$-set agreement task one would say that each process proposes
a value, and decides a value, such that (validity) 
 a decided value has been proposed,
and (agreement) at most $k$ different values are decided. A formal definition
of when an execution satisfies a task is derived next.
A task $T$ has only one operation, ${\sf task}()$, which  process $\id_i$
may call with  value $x_i$, if $(\id_i,x_i)$ is a vertex of $\m{I}$.
The operation ${\sf task}(x_i)$ may return  $y_i$ to the process,
if $(\id_i,y_i)$ is a vertex of $\m{O}$.
Let $E$ be an execution where each process calls ${\sf task}()$  once.
Then, $\sigma_{E}$ denotes the simplex
containing all input vertices in $E$, namely, if in $E$ there is an invocation
of ${\sf task}(x_i)$ by process $\id_i$ then $(\id_i,x_i)$ is in  $\sigma_{E}$.
Similarly, 
$\tau_{E}$ denotes the simplex  containing all output vertices  in $E$,
namely, $(\id_i,y_i)$ is in $\tau_{E}$ iff there is a response $y_i$ to a process $\id_i$ in $E$.
We say that
$E$ \emph{satisfies}  task $T = \langle I, O, \Delta \rangle$
if for every prefix $E'$ of $E$, it holds that $\tau_{E'} \in \Delta(\sigma_{E'})$.
It is necessary to  consider 
all prefixes of an execution, to  
 prevent anomalous executions that globally
seem correct, but in a prefix a
process predicts future invocations, as in the execution of the 
$validity$ task in Figure~\ref{fig-validityMislead}.\footnote{This prefix requirement 
has been implicitly considered in the past by stating that
an algorithm solves a task if any of its executions agree with the specification of the task.}

\begin{figure*}
\begin{center}
	\epsfxsize=4.50in
	\epsfbox{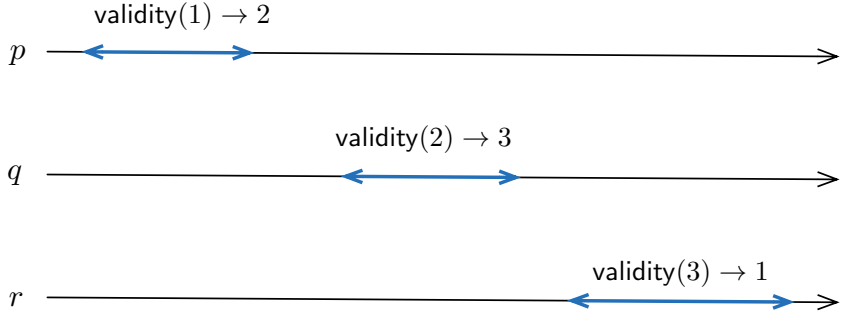}%
    \end{center}
\caption{An execution that does not satisfy the $validity$ task.}
\label{fig-validityMislead}
\end{figure*}

\paragraph{From tasks to interval-sequential objects}
A task is a very compact way of specifying a distributed problem that
is capable of describing allowed behaviours for certain concurrency patterns,
and indeed it is hard to understand what exactly is the problem being specified.
The following theorem (with its proof) provides an automata-based 
representation of a task,  explaining which outputs may be
produced in each execution, as permitted by $\Delta$.

\begin{figure*}
\begin{center}
	\epsfxsize=6.5in
	\epsfbox{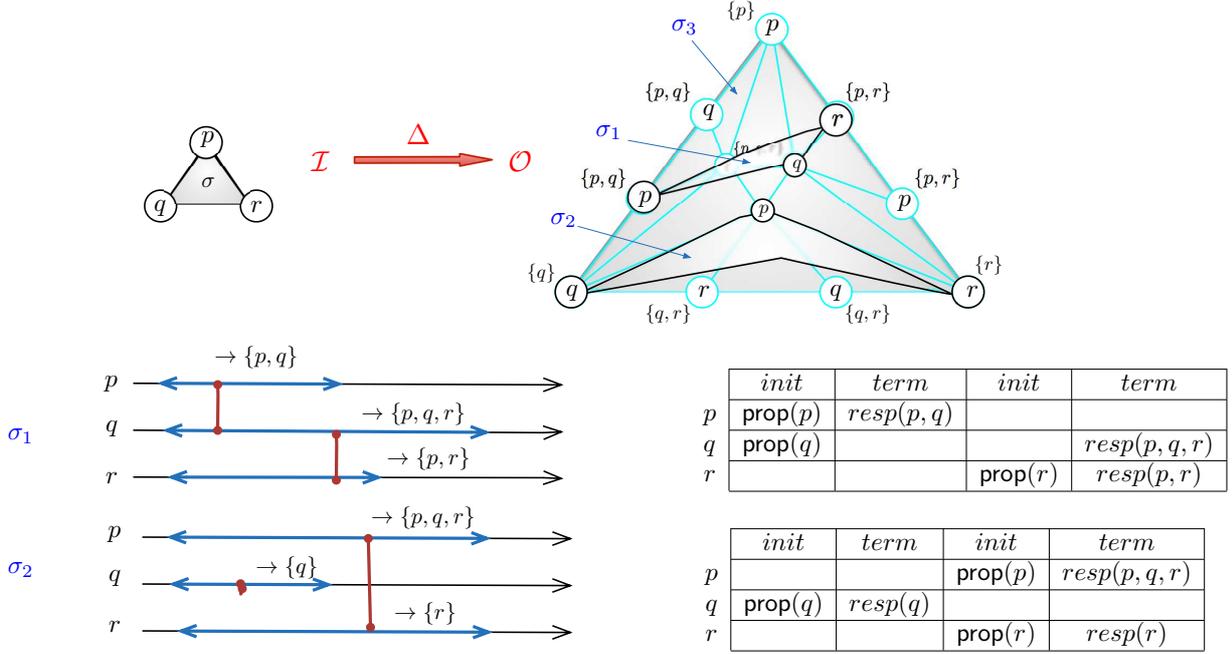}%
    \end{center}
\caption{Two  special output simplexes $\sigma_1,\sigma_2$, and 
interval-linearizations of two executions with corresponding outputs}
\label{fig-simplexInterp}
\end{figure*}

\begin{theorem}
\label{theo-from-tasks-to-objects}
For every task $T$, there is an interval-sequential object $O_T$ such that 
an execution $E$ satisfies $T$ if and only if 
it is interval-linearizable with respect to $O_T$.
\end{theorem}

%

To give an intuition of the insights in the proofs of this theorem, 
consider the immediate snapshot task (Figure~\ref{fig:fig-is-task}).
A simple case is the output simplex $\sigma_4$ in the center of the output 
complex,
where the three processes output $\{p,q,r\}$. It is simple, because this
simplex does not intersect the boundary. Thus, it can be produced as output 
only when all three operations are concurrent.
More interesting is output simplex $\sigma_3$, where they also may 
run concurrently,
but in addition, the same outputs may be returned in a fully 
sequential execution,
because $\sigma_3$ intersects both the 0-dimensional and the 1-dimensional 
boundary of the output complex. In fact $\sigma_3$ can also be produced
if $p,q$ are concurrent, and later comes $r$, because 2 vertices of $\sigma_3$
are  in $\Delta(p,q)$.
Now, consider the two more awkward output simplexes 
$\sigma_1,\sigma_2$ in $\Delta(\sigma)$
added to the immediate-snapshot output complex
 in Figure~\ref{fig-simplexInterp}, where
 $\sigma_1=\{(p,\{p,q\}), (q, \{p,q,r\})  ,(r, \{p,r\})\}$, and
 $\sigma_2=\{(p,\{p,q,r\}), (q, \{q\})  ,(r, \{r\})\}$.
 At the bottom of the figure, two executions and 
their interval-linearizations are shown, though there are more 
executions that are interval-linearizable 
and can produce $\sigma_1$ and $\sigma_2$.
Consider $\sigma_2$, which is in $\Delta(\sigma)$. Simplex $\sigma_2$ 
has a face, $\{q\}$,
in $\Delta(q)$, and another face, $\{r\}$ in $\Delta(r)$. This specifies
a different behavior from the output simplex in the center, than does not
intersect with the boundary. 
Since $\Delta(\{q\})=\{q\}$, it is OK for $q$ to return
$\{q\}$ when it invokes and returns before the others invoke.
Now, since 
$\{\{p,q,r\},q,r\} \in \Delta(\{p,q,r\})$ then 
 it is OK for $r$ to return $\{r\}$ after everybody has invoked. 
Similarly, since 
$\{\{p,q,r\},q,r\} \in \Delta(\{p,q,r\})$,  $p$ can return $\{\{p,q,r\},q,r\}$.
The main observation here is that the structure of the mapping $\Delta$ encodes
the interval-sequential executions that can produce the outputs in a given 
output simplex.
In the example, $\Delta$ precludes the possibility that in a sequential 
execution the processes outputs the values in $\sigma_1$, since $\Delta$
 specifies no process can decide without seeing anyone else.

\paragraph{From one-shot interval-sequential objects to tasks} 
The converse of Theorem~\ref{theo-from-tasks-to-objects} is not true. 
Lemma~\ref{lemm-queueNoTask} shows that even some sequential objects,
such as queues, cannot be represented as a task.
Also, recall that  there are tasks with no set-sequential specification.
Thus, both tasks and set-sequential objects are 
 interval-sequential objects, but they are incomparable.

\begin{lemma}
\label{lemm-queueNoTask}
There is a sequential one-shot object $O$ such that there is no
task $T_O$, satisfying that an execution $E$ is linearizable with respect to $O$
if and only if $E$ satisfies $T_O$ (for every $E$).
\end{lemma}

We have stablished that tasks have strictly less expresive power than
interval-sequential one-shot objects, however, a slight modification of 
the notion of tasks allows to equate the power of both approaches for specifying distributed one-shot problems. 
Roughly speaking, tasks cannot model interval-sequential 
objects because they do not have a mechanism to encode the state of an object.
The extension we propose below allows to model states.

In a \emph{refined} task $T = \langle \m{I}, \m{O}, \Delta \rangle$,
 $\m{I}$ is defined as usual and each output vertex of $\m{O}$ has the form $(id_i. y_i, \sigma_i')$
where $id_i$ and $y_i$ are, as usual, the ID of a process and an output value,
and $\sigma_i'$ is an input simplex called the \emph{set-view} of $id_i$.
The properties of $\Delta$ are maintained and in addition it satisifies the following:
for every $\sigma \in \m{I}$, for every $(id_i, y_i, \sigma_i') \in \Delta(\sigma)$,
it holds that $\sigma_i' \subseteq \sigma$.
An execution $E$ \emph{satisfies} a refined task $T$ 
if for every prefix $E'$ of $E$, it holds that $\Delta(\sigma_{E'})$ contains
the simplex
$\{ (id_i, y_i, \sigma_{i \, E''}) :  (id_i, y_i) \in \tau_{E'} \, \wedge  
\hbox{ $E''$ (which defines ${\sigma_i}_{E''}$) is the shortest prefix of $E'$ containing the response $(id_i,y_i)$} \}$.

We stress that, for each input simplex $\sigma$, 
for each output vertex $(id_i, y_i, \sigma_i) \in \Delta(\sigma)$, $\sigma_i$
is a way to model distinct output vertexes in $\Delta(\sigma)$ 
whose output values (in $(id_i, y_i)$) are the same,
then a process that outputs that vertex does not actually output $\sigma_i$.
In fact, the set-view of a process $id_i$ corresponds to the set of invocations
that precede the response $(id_i, y_i)$ to its invocation in a given execution
(intuitively, the invocations that a process ``sees'' while computing its output value ).
Set-views are the tool to encode the state of an object.
Also observe that if $E$ satisfies a refined task $T$,
then the set-views behave like snapshots:
1) a process itself (formally, its invocation) appears in its set-view and
2) all set-view are ordered by containment (since we assume $E$ is well-formed).

As already mentioned, interval-sequential objects and refined tasks
have the same ability to specify distributed one-shot problems, 
as the following theorems show. The proof of Theorem~\ref{theo-from-ref-tasks-to-objects} is essentially the same
as the proof of Theorem~\ref{theo-from-tasks-to-objects}.

\begin{theorem}
\label{theo-from-objects-to-tasks}
For every one-shot interval-sequential object $O$ with a single total operation,
there is a refined task $T_O$ such that any execution $E$ is interval-linearizable with respect to $O$
if and only if $E$ satisfies $T_O$.
\end{theorem}

\begin{theorem}
\label{theo-from-ref-tasks-to-objects}
For every refined task $T$, there is an interval-sequential object $O_T$ such that 
an execution $E$ satisfies $T$ if and only if 
it is interval-linearizable with respect to $O_T$.
\end{theorem}

%
%
%
%

\section{Conclusion}
\label{sec:conclusion}
We have proposed the notion of an {\it interval-sequential object}, specified
by a state machine similar to the ones used for sequentially specified
objects, except that transitions are labeled with sets of invocations
and responses, instead of operations, to represent operations
that span several consecutive transitions. Thus, in a
state an invocation might be pending. 
The corresponding consistency condition is {\it interval-linearizability}.
If an execution is interval-linearizable for an object  $X$, its
invocations and responses can be organized, respecting real-time, in a way 
that they can be executed through the automaton of $X$.
Thus, contrary to the the case of linearizability where to linearize an execution
one has to find unique linearization points, 
for interval-linearizability one needs to identify an interval of time for each operation, and the intervals might overlap.
We have shown that by going from linearizability to
interval-linearizability one does not sacrifice the properties of being 
local and non-blocking.

We have discovered that interval-sequential objects have strictly more
expressive power than tasks. Any algorithm that solves a given task
is interval-linearizable with respect to the interval-sequential object that corresponds to the task,
however, there are one-shot objects that cannot be expresses as tasks.
We introduced the notion of \emph{refined tasks} and prove that interval-sequential objects
and refined tasks are just two different styles, equally expressive,  of specifying concurrent 
one-shot problems, the first operational, and the second static. 
This brings benefits from each style
to the other, and finally provides a common framework to think about
linearizability, set-linearizability, interval-linearizability, and tasks.


There are various directions interesting to pursue further.
In the domain of concurrent specifications, there is interest
in comparing the expressive power of several models of concurrency,
e.g.~\cite{G2006},
and as far as we know, no model similar to ours has been considered.
Higher dimensional automata~\cite{Pratt91},  the most expressive
model in~\cite{G2006}, seems related to set-linearizability.
Also, several papers explore  partial order semantics of programs.
More flexible notions of linearizability, relating  two
arbitrary sets of histories appear in~\cite{FORY2010}, but without stating
a compositionality result, and without an automata-based formalism.
 However it is worth exploring this direction further, as it establishes 
that linearizability implies
observational refinement, which usually entails compositionality 
(see, e.g.,~\cite{GMY12}). Also, it would be interesting to consider
that  in this semantics two events in a single trace can be related in three ways: definitely dependent, definitely concurrent or unrelated.

Several versions of non-determinism
were explored in~\cite{CRR13}, which could be understood through
the notions in this paper.
Also, it would be interesting to consider multi-shot task versions that
correspond to interval-sequential objects, as well as the implications
of the locality property. 

As observed in~\cite{GKL12}, devising linearizable objects can be 
very difficult,
requiring complex algorithms to work correctly under general circumstances,
and often resulting in bad average-case behavior. 
Programmers thus optimize algorithms to handle common scenarios more 
efficiently.
The authors propose \emph{speculative linearizability} to
simplify the design of efficient yet robust linearizable protocols. 
It would be interesting to see if similar techniques can be used for 
interval-specifications of concurrent objects proposed here, 
and if our more generic composability proof
sheds light on the composability result of~\cite{GKL12}.

Often concurrent data structures shared  require linear 
worst case time to perform a single instance
of an operation in any non-blocking implementation~\cite{EHS12}, else,
they are not linearizable e.g.~\cite{HSW96}. Thus, concurrent specifications,
such as interval-linearizable objects open possibilities of sub-linear time
implementations.

Finally, Shavit~\cite{shavit11} summarizes beautifully the common knowledge 
state that
 ``it is infinitely easier and more intuitive for us humans to specify how 
abstract data structures behave in a sequential setting.
Thus, the standard approach to arguing the safety properties of a concurrent 
data structure is to specify the structure's properties sequentially, and find 
a way to map its concurrent executions to these `correct' sequential ones."
We hope  interval-linearizability opens the possibility of facilitating
reasoning about concurrent specifications, when no sequential specifications
are appropriate.


\newpage
\setcounter{page}{1}
\pagenumbering{roman}

\section*{Acknowledgments}
A. Casta\~neda was partially supported by a PAPIIT-UNAM research grant.
S. Rajsbaum was partially supported by a PAPIIT-UNAM, and a LAISLA 
Mexico-France research grant.
M. Raynal was partially supported by the  French  ANR  
project DISPLEXITY devoted to  computability and  complexity in distributed 
computing, and the Franco-German ANR project DISCMAT devoted to connections 
between mathematics and distributed computing.



\appendix
\section{Linearizability}
\label{sec:linearizability}

A \emph{sequential object} $O$ 
is a (not necessarily finite)
Mealy state machine $(Q,{Inv},{Res},\delta)$  whose output values 
are determined both by its current state $s\in Q$ and the current
input $I\in {Inv}$. 
If $O$ is in state $q$ and it receives as input an invocation $in\in Inv$
by process $p$,
then, if  $\delta(q,inv)=(r,q')$, the meaning is that $O$ may return
the  response $r$ to the invocation $inv$ by process $p$, and move to state $q'$.  
Notice that  the   response $r$ has to be to the invocation by $p$,
but there  may be several possible responses (if the object is non-deterministic).
Also, it is convenient to require that the object is \emph{total}, meaning that
for any state $q$, $\delta(q,I)\neq \emptyset$, for all $I\in  {Inv}$.

Considering any object 
defined by a sequential specification on total operations, 
linearizability~\cite{HW90} generalizes 
the notion of an atomic read/write object formalized in~\cite{L86-a,M86}, and 
encountered in virtual memory-based distributed systems~\cite{LH89}.

Intuitively, an execution is linearizable if it could have been produced 
by multiplexing the processes on a single processor. 
This definition considers complete histories. If the execution is partial, 
an associated complete execution can be  defined as follows. 
The local execution $\widehat{H}|i$ of  each process $p_i$ for which the  
last operation is pending (i.e., $p_i$ issued an invocation and there no 
matching response event),  is completed with a  response matching the
invocation event. Thus, it may be  possible to associate  different complete 
histories with a  given partial execution.

An execution  $E$ is linearizable if there is and extension $\overline E$ of $E$  
and a sequential execution 
$\widehat{S}$ such that:
\begin{itemize}
\item  $comp(\overline E)$ and $\widehat{S}$ are equivalent (no process can distinguish between 
 $comp(\overline E)$ and $\widehat{S}$). 
\item  $\widehat{S}$ is legal (the specification of each object is respected). 
\item The total order $\widehat{S}$ 
respects the partial order $\OP$ associated to $comp(\overline E)$
 (any two operations ordered in $\OP$ are ordered the same way in  $\widehat{S}$). 
\end{itemize}

As shown in~\cite{HW90}, the linearizability  
consistency condition  has the ``composability'' property (called ``locality'' 
in~\cite{HW90}), which states that a computation  $E$ is linearizable 
if and only if,   for each of its objects $X$,  $E_|X$ is linearizable. 

\section{Additional details about the write-snapshot task}
\label{sec:proof-write-snapshot}

Recall that in the {write-snapshot task} 
the   ${\sf write}()$ and  ${\sf snapshot}()$ operations are merged to 
define a single operation denoted ${\sf write\_snapshot}()$. 
It satisfies the  self-inclusion and containment properties.
Notice that the \emph{immediate snapshot} task~\cite{BG93}
which motivated Neiger to propose set-linearizability~\cite{N94}
 is a write-snapshot  which additionally 
satisfies the following immediacy property: 
$\forall~i,j:~ [(\langle j,- \rangle set_i)\wedge
(\langle i,- \rangle set_j)] \Rightarrow  (set_i=set_j)$.

For completeness and comparison, we inlcude the following proof 
in the usual, somewhat informal style, of the correctness of the
write-snapshot algorithm.
To simplify the presentation we suppose that the value 
written by $p_i$ is $i$, and the pair $\langle i,v_i\rangle$ 
is consequently denoted $i$.

\begin{theorem}
\label{theo:write-snapshot}
The algorithm of Figure~{\em\ref{fig:algorithm-write-snapshot}}
wait-free implements write-snapshot.   
\end{theorem}

\begin{proofT}
Let us first show that  the invocation of  ${\sf write\_snapshot}()$  by any 
process $p_i$ terminates. As there is a bounded number of processes, and a 
process invokes  ${\sf write\_snapshot}()$ at most once, it follows that a 
process can be forced to execute at  most $(n-1)$ double collects, 
and the termination follows. 

The self-inclusion property follows immediately from line~\ref{WS-01}, 
and the fact that no value is ever withdrawn from the array $\MEM$.

To prove the containment property, let us consider two processes 
$p_i$ and $p_j$, which return $set_i$ and $set_j$, respectively. 
Let us first consider $p_i$. As it returns $set_i$, we have 
$set_i=old_i=new_i$ where $new_i$ corresponds to the last
asynchronous read of $\MEM[1..n]$ by $p_i$,  and $old_i$ corresponds to
the previous asynchronous read of $\MEM[1..n]$. 
Let $\tau[old_i]$ the time at which terminates the read  of $\MEM[1..n]$
returning $old_i$,  and  $\tau[new_i]$ the time at which starts the read  of 
$\MEM[1..n]$ returning $new_i$.
As $old_i=new_i$, it follows that there is a time $\tau_i$, 
such that  $\tau[old_i] \leq \tau_i \leq \tau[new_i]$ and, due to the 
termination predicate of line~\ref{WS-05}, the set of 
non-$\bot$ values of $\MEM[1..n]$ at time   $\tau_i$ is equal to $set_i$. 

The same applies to $p_j$, and there is consequently a time  $\tau_j$ 
at which the set of non-$\bot$ values of $\MEM[1..n]$ is equal to $set_j$. 
Finally, as (1) $\tau_i\leq \tau_j$ or  $\tau_i>  \tau_j$, 
and (b)  values written in $\MEM[1..n]$ are never withdrawn, it follows 
that we necessarily have $set_i\subseteq set_j$ or $set_j\subseteq set_i$.
\renewcommand{\toto}{theo:write-snapshot}
\end{proofT}

\paragraph{A finite state automaton describing the behavior of a write-snapshot object}
The non-deterministic automaton of Figure~\ref{fig:write-snapshot-automaton} 
describes in an abbreviated form
all the possible  behaviors of a write-snapshot object in a system 
of three processes $p$, $q$, and  $r$. To simplify
the figure, it is assumed that a process $p_i$ proposes $i$.
Each edge correspond to an invocation of  ${\sf write\_snapshot}()$, and  
the list of integers $L$ labeling a transition edge means that the 
corresponding invocation of  ${\sf write\_snapshot}()$ is by one of 
the processes $p_i$ such that $i\in L$.  
The value returned by the object is $\{ L\}$.
Thus, for the linearization of the execution in 
Figure~\ref{fig:simple-counterEx},
the path in the automaton goes through states 
$\emptyset,\{1,2\},\{1,2\},\{1,2,3\}$.

Any path starting from the initial empty state, and in which a process index 
appears at most once, defines an execution of the write-snapshot task
that does not predict the future.
Moreover if,  when it executes, a process proceeds from the automaton state
$s_1$ to the state $s_2$, the state $s_2$ defines the tuple of values
output by its invocation of ${\sf write\_snapshot}()$.

\begin{figure*}[th]
\centering{
\hspace{-0.2cm}
\scalebox{0.35}{\input{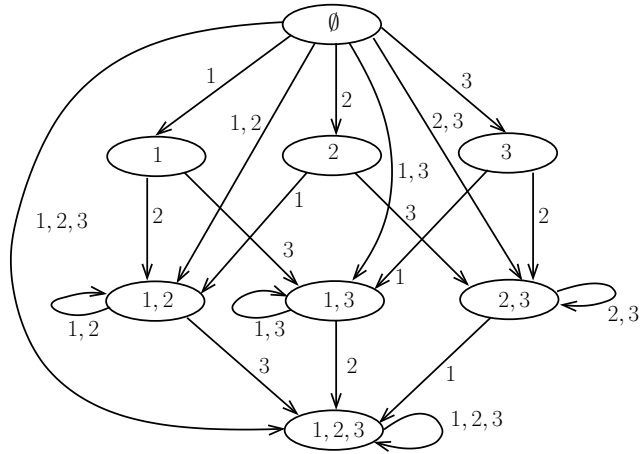}}
\caption{A non-deterministic automaton for a write-snapshot object}
\label{fig:write-snapshot-automaton}
}
\end{figure*}

\section{Additional discussion and examples of linearizability limitations}
\label{app:addDiscExam}

\subsection{Additional examples of tasks with no sequential specification}
\label{app:addExampNoSeq}
Several tasks have been identified 
that are problematic for dealing with
them through linearizability. The problem is that they
 do not have a natural sequential specification.
One may consider
linearizable implementations of restricted sequential specifications,
where if two operations occur concurrently, one is linearized before
the other. Thus, in every execution, always there is a first operation.  
In all cases we discuss below, such an implementation would provably
be of a more powerful object.

An \emph{adopt-commit} object~\cite{G98} is a one-shot shared-memory object
useful  to implement round-based protocols for set-agreement and consensus.
It supports a single operation, ${\sf adopt\_commit}()$.
 The  result  of  this  operation  is  an  output  of  the form  $(commit,v)$ 
or 
 $(adopt,v)$, where  the  second  component  is  a  value  from  this
set and the 1st component  indicates whether the process
should decide value $v$ immediately or adopt it as its preferred value 
in later rounds
of the protocol.  It has been shown to be equivalent to the 
\emph{conflict detection} object~\cite{AE14}, which  supports  a  single  
operation, ${\sf check}()$. 
It returns true or false, and has the following two properties:
In any execution that contains a ${\sf check}(v)$ operation and a 
${\sf check}(v')$ operation with $v\neq v'$, at least one of these 
operations returns true. In any execution in which all check
operations have the same input value, they all return false.
As observed in~\cite{AE14}  neither  adopt-commit  objects  nor  
conflict  detectors have  sequential specification.
A deterministic linearizable implementation of an adopt-commit 
object gives rise to a deterministic implementation of consensus, 
which does not exist. Similarly, the first check operation linearized 
in any execution of a conflict detector must return false and 
subsequent check operations  with  different  inputs  must  return true,  
which  can  be  used  to  implement  test-and-set,  for  which  no 
 deterministic implementation from registers exists.

In the \emph{safe-consensus}  problem of~\cite{AGL09}, the agreement condition
of consensus is retained, but the validity condition  is weakened as follows: 
if the first process to invoke it returns before any other process invokes it, 
then it outputs its input; otherwise the consensus output can be arbitrary,
 not even the input of any process.
There is no sequential specification of this problem, because in any 
sequential specification, the first process to be linearized would 
obtain its own proposed value. See Appendix~\ref{app:safe-cons-and-abort}.

Two examples that motivated Neiger are the following~\cite{N94}.
In the \emph{immediate snapshot} task~\cite{BG93},  there is a single
operation ${\sf Immediate\_snapshot}()$, such that
 a snapshot occurs immediately after a read. Such executions
play an important role in distributed computability~\cite{AR02,BG93,SZ00}.
There is no sequential specification of this task. One may consider
linearizable implementations of restricted immediate snapshot behavior,
where if two operations occur concurrently, one is linearized before
the other, and where the first operation does not return the value by the 
second. But such an implementation would provably
be of a more powerful object (immediate snapshots  can be implemented
wait-free using only read/write registers), that could simulate test-and-set.

The other prominent example  exhibited  in~\cite{N94} is the 
\emph{$k$-set agreement} task~\cite{C93},
where processes agree on at most $k$ of their input values. 
Any linearizable
implementation restricts the behavior of the specification, because some
process final value would have to be its own input value. This would be
an artifact imposed by linearizability. Moreover, there are implementations
of set agreement with executions where no process chooses its own initial value.

\subsection{Splitting operations to model concurrency}
\label{app:splitting}
One is tempted to separate an operation into two, an invocation and a response,
to specify the effect of concurrent invocations. 
Consider two operations of an object,
${\sf op}_1()$ and ${\sf op}_2()$, such that each one is invoked with a 
parameter and can return a value. 
Suppose we want to specify how the object behaves
when both are invoked concurrently. We can separate each one into
two operations, ${\sf inv\_op}_i()$ and ${\sf resp\_op}_i()$. When 
a process wants to invoke ${\sf op}_i(x)$, instead it first
invokes ${\sf inv\_op}_i(x)$, and once the operation terminates, it invokes
${\sf resp\_op}_i()$, to get back the output parameter.  Then 
a sequential specification can define what the operation returns
when the history is ${\sf inv\_op}_1(x_1),{\sf inv\_op}_2(x_2),
{\sf resp\_op}_1(),{\sf resp\_op}_2()$.

$k$-Set agreement is easily transformed into an object with a sequential
specification, simply by accessing it through two different operations, 
one that deposits a value into the object and another that returns one
of the values in the object. Using a non-deterministic specification that
remembers which values the object has received so far, and which
ones have so far been returned, one captures the behavior that at most 
$k$ values are returned, and any of the proposed values can be returned.
This trick can be used in any task.

Separating an operation into a proposal operation and a returning operation
 has several problems. First, the program is forced to 
produce two operations, and wait for two responses. There is a 
consequent loss of clarity in the code of the program, in addition to a loss 
in performance, incurred by  a two-round trip delay. Also, the intended meaning
of linearization points is lost; an operation is now linearized at \emph{two}
linearization points.
Furthermore, the resulting object may  provably \emph{not} be the same.
A phenomenon that has been observed several times 
(see, e.g., in~\cite{CR14,GR10,RRT08}) 
is that the power of the object can be increased, if one is allowed to 
invoke another object in between the two operations. 
Consider a test-and-set object that returns either 0 or 1, and
the write-snapshot object. It is possible to solve consensus
among 2 processes with \emph{only one} snapshot object and 
one test-and-set object only if 
it  is allowed to invoke test-and-set
in between the write and the snapshot operation.
Similarly, consider a safe-consensus object instead of the test-and-set object.
 If one is allowed to invoke in between the two operations of  
write-snapshot a safe-consensus object, then one can solve consensus
more efficiently~\cite{CR14}.

\paragraph{The object corresponding to a  task with two operations}
Let $T$ be a task $(\m{I},\m{O},\Delta)$.
 We will model $T$ as a sequential object $O_T$ in which each process can invoke two 
 operations, {\sf set} and {\sf get}, in that order. The idea is that {\sf set} communicates to $O_T$
 the input value of a process, while {\sf get} produces an output value to a process.
 Thus, the unique operation of $T$ is modelled with two operations.
 The resulting sequential object is non-deterministic.
 
 We define $O_T$.
 The set of invocations and responses  are the following:
 $$Inv(O_T) = \{ {\sf set}(p_i, in_i) \, | \, (p_i, in_i) \in {\cal I} \} \cup \{ {\sf get}(p_i) \, | \, p_i \in \Pi \}$$
$$Res(O_T) = \{ {\sf set}(p_i, in_i):OK \, | \, (p_i, in_i) \in {\cal I} \} \cup \{ {\sf get}(p_i): out_i \, | \, p_i \in \Pi \, \wedge \, (p_i, out_i) \in {\cal O}\}$$
 The set of states of $O_T$ is $Q = \{ (\sigma, \tau) | \sigma \in {\cal I} \, \wedge \, \tau \in \Delta(\sigma) \}$.
 Intuitively, a set $(\sigma, \tau)$ represents that the inputs and output $O_T$ knows at that state are $\sigma$ and $\tau$. 
 The initial state of is $(\emptyset, \emptyset)$. We define $\delta$ as follows.
 Let $(\sigma, \tau)$ and $(\sigma', \tau')$ be two states of $O_T$.
 Then,
 
 \begin{itemize}
 
 \item If $\tau = \tau'$, $\sigma \neq \sigma'$ and $\sigma' = \{\sigma \cup  (p_i, in_i) \}\in\m{I}$,
 then $\delta((\sigma, \tau), {\sf set}(p_i, in_i))$ contains the tuple 
 $({\sf set}(p_i, in_i):OK, (\sigma', \tau'))$.
 
 \item If $\sigma = \sigma'$, $\tau \neq \tau'$ and $\tau' =\{ \tau \cup  (p_i, out_i) \}\in\Delta(\sigma)$,
 then $\delta((\sigma, \tau), {\sf get}(p_i))$ contains the tuple 
 $({\sf get}(p_i):out_i, (\sigma', \tau'))$.
 \end{itemize}
 
Note that for every sequential execution $\widehat S$ of $O_T$,
 it holds that $\tau_{\widehat S} \in \Delta(\sigma_{\widehat S})$,
 where $\sigma_{\widehat S}$ is the input simplex containing every
 input vertex in $\widehat S$ and, similarly, $\tau_{\widehat S}$
 is the output simplex containing every output simplex in~$\widehat S$.

%
%

\subsection{Validity and Safe-consensus objects}
\label{app:safe-cons-and-abort}
We first discuss the validity object with abort, and then the safe-consensus object.

\subsubsection{Validity with abort object}
\label{app:val-abort}

An interval-sequential object can be enriched with an abort
operation that takes effect only if a given number of processes
request an abort concurrently. Here we
describe the example of  Section~\ref{sec:valObj} in more detail,
that extends the validity object with an ${\sf abort}$ operation that should
be invoked concurrently by at least $k$ processes.
As soon as at least $k$ processes concurrently invoke ${\sf abort}$
the object will return from then on ${\sf aborted}$ to every operation.
Whenever less than  $k$ processes are concurrently invoking ${\sf abort}$,
the object may return ${\sf NotAborted}$ to any pending ${\sf abort}$.
An example appeared in Figure~\ref{fig:fig-validityAbort}, for $k=2$.
Another example is
in Figure~\ref{fig:fig-validityAbort1}, where it is shown that even though
there are two concurrent ${\sf abort}$  operations, they do not take effect
because they are not observed concurrently by the object.
This illustrates why this paper is only about safety properties, 
the concepts here cannot enforce liveness. There is no way of guaranteeing
that the object will abort even in an execution where all processes issue
 ${\sf abort}$ at the same time, because the operations may be 
 executed sequentially.
\begin{figure*}
\begin{center}
	\epsfxsize=4.0in
	\epsfbox{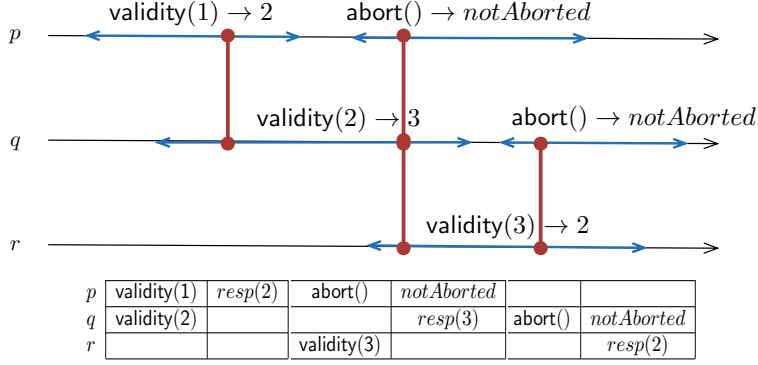}%
    \end{center}
\caption{An execution of a Validity-Abort object (3)}
\label{fig:fig-validityAbort1}
\end{figure*}

The $k$-\emph{validity-abort} object is formally specified as an interval-sequential object
by an automaton, that can be invoked by either  ${\sf propose}(v)$ 
or ${\sf abort}$, and it responds with either  ${\sf resp}(v)$
 or ${\sf aborted}$ or ${\sf NotAborted}$.
Each state $q$ is labeled with three values: $q.vals$ is the set
of values that have been proposed so far,  $q.pend$ is the set of
processes with pending invocations,
and $q.aborts$ is the set of processes with pending ${\sf abort}$. 
The initial state $q_0$
has $q_0.vals=\emptyset$, $q_0.pend=\emptyset$ and $q_0.aborts=\emptyset$.
If $in$ is an invocation to the object different from ${\sf abort}$, 
let $val(in)$ be the proposed value,
and if $r$ is a  response from the object, let $val(r)$ be the responded value.

For a set of invocations $I$ (resp. responses $R$) $vals(I)$ denotes
the proposed values in $I$ (resp. $vals(R)$). Also, $aborts(I)$
denotes the set of processes issuing an ${\sf reqAbort}$ in $I$,
and $notAborted(R)$ is the set of processes getting ${\sf notAborted}$ in $R$.

The transition relation  $\delta(q,I)$ contains all pairs $(R,q')$ such that:
\begin{enumerate}
\item 
If $r\in R$ then $id(r)\in q.pend$ or there is an $in\in I$ with $id(in)=id(r)$,
\item 
If $(r={\sf resp}(v)\in R$  or ${\sf notAborted}\in R)$ then ${\sf aborted}\not\in R$, 
\item
If $r={\sf resp}(v)\in R$ then
$val(r)=v\in q.vals$ or there is an $in\in I$ with 
$val(in)=val(r)$, 
 \item
 If ${\sf notAborted}\in R$, then
 $0< |q.aborts|+|aborts(I)| < k$
\item 
If $|q.aborts|+|aborts(I)| \geq k$ then
 ${\sf aborted}\in R$.
\item 
$q'.vals= q.val\cup vals(I)$,
 $q'.pend= (q.pend \cup ids(I))\setminus ids(R)$, and\\
 $q.aborts = (q.aborts\cup aborts(I)) \setminus notAborted(R) $
\end{enumerate}

\subsubsection{Safe-consensus}
\label{app:safe-cons}
Recall that  the \emph{safe-consensus}  problem of~\cite{AGL09}, 
is similar to consensus. The agreement condition
of consensus is retained, but the validity condition  is weakened as follows: 
if the first process to invoke it returns before any other process invokes it, 
then it outputs its input; otherwise the consensus output can be arbitrary,
 not even the input of any process.
As noticed in Section~\ref{app:addExampNoSeq}, 
there is no sequential specification of this problem.

See Figure~\ref{fig:fig-intAutomata} for part of the automata
corresponding to safe-consensus,
and examples of interval executions in Figure~\ref{fig:interval-execution}.

\vspace{0.4cm}
\begin{figure*}
\begin{center}
	\epsfxsize=5.5in
	\epsfbox{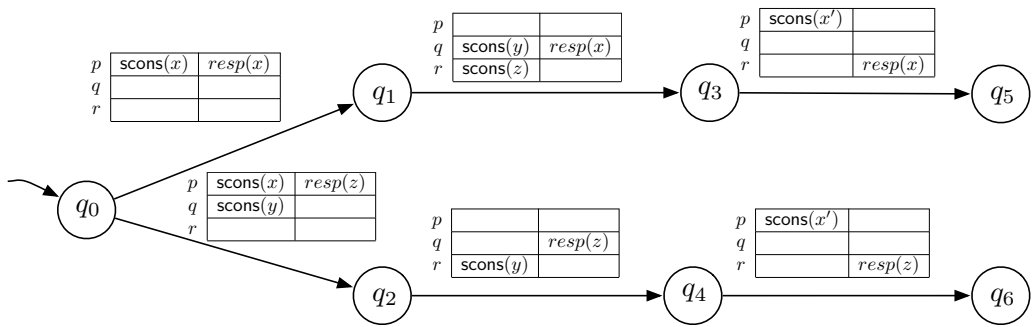}%
    \end{center}
\vspace{-1cm}
\caption{Part of an interval-sequential automaton of safe-consensus}
\label{fig:fig-intAutomata}
\end{figure*}

\vspace{1.5cm}
\begin{figure*}[th]
\centering{
\hspace{0.0cm}
Interval execution $\alpha_1$\\
$\begin{array}{c|c|c|c|c|c|c|}
 \cline{2-7}
   & {init} & {term} & {init} & {term} & {init} & {term}  \\
 \cline{2-7}
 p &  {\sf scons}(x)	&   resp(x)	&  {\,}	& {\,}		& {\,} {\sf scons}(x')	& {\,}  \\
 \cline{2-7}
q &  {\,}    		&   {\,}	& {\sf scons}(y) & {\,} resp(x) & {\,}		& {\,}   \\
 \cline{2-7}
r &   {\,} 		&   {\,}	&  {\sf scons}(z) & {\,}	 & {\,}		& resp(x)	  \\
 \cline{2-7}
\end{array}
$\\
$\,$\\
$\,$\\
Interval execution $\alpha_2$\\
$\begin{array}{c|c|c|c|c|c|c|}
 \cline{2-7}
   & {init} & {term} & {init} & {term} & {init} & {term}  \\
 \cline{2-7}
 p &  {\sf scons}(x)	&   resp(z)	&  {\,}	& {\,}		& {\,} {\sf scons}(x')	& {\,}  \\
 \cline{2-7}
q &  {\sf scons}(y)  		&   {\,}	& {\,}   & {\,} resp(z) & {\,}		& {\,}   \\
 \cline{2-7}
r &   {\,} 		&   {\,}	&  {\sf scons}(z) & {\,}	 & {\,}		& resp(z)	  \\
 \cline{2-7}
\end{array}
$
\caption{Examples of interval-executions for safe-consensus}
\label{fig:interval-execution}
}
\end{figure*}

\section{Tasks}
\label{app:tasks}

\subsection{Basic definitions}
A  task is the basic distributed equivalent of a function,
defined by a set of inputs to the processes and for each (distributed) 
input to the processes, a set of legal (distributed) outputs of the processes, e.g.,~\cite{HKRbook}.
In an algorithm designed to solve a task, each process starts with a private input
value and has to eventually decide irrevocably on an output value. 
A process $p_i$ is initially not aware of the inputs  of other processes.
Consider an execution where only a subset of $k$ processes
participate; the others crash without taking any steps. 
A  set of pairs
$s=\{(\id_1,x_1),\dots,(\id_k,x_k)\}$ is used to denote the input values,
or output values, in the execution, where  $x_i$ 
denotes the value of the process with identity $\id_i$, either an input
value, or a output value. 

A set $s$ as above is called a \emph{simplex}, and
if the values are input values, it is an \emph{input simplex},
if they are output values, it is an \emph{output simplex}.
The elements of $s$ are called \emph{vertices}. An \emph{input vertex} $v=(\id_i,x_i)$
represents the initial state of process $\id_i$, while an \emph{output vertex}
represents its decision.
The \emph{dimension} of a simplex $s$ is $|s|-1$,
and it is  \emph{full} if it contains $n$ vertices, one for each process.
A subset of a simplex is called a \emph{face}.
 Since any number of processes may crash, simplexes of all dimensions
 are of interest, for taking into account executions
where only processes in the simplex participate. Therefore, the set of 
possible input simplexes forms a \emph{complex} because its sets are
closed under containment. Similarly, the set of possible output simplexes
also form a complex.

More generally, a \emph{complex}
$\m{K}$ is a set of vertices $V(\m{K})$, and a family of finite,
nonempty subsets of $V(\m{K})$, called \emph{simplexes}, satisfying:
(1) if $v\in V(\m{K})$ then $\{v\}$ is a simplex, and (2) if $s$ is a
simplex, so is every nonempty subset of $s$. 
The dimension of $\m{K}$ is the largest
dimension of its simplexes, and $\m{K}$ is \emph{pure} of dimension
$k$ if every simplex belongs to a $k$-dimensional simplex. In
distributed computing, the  simplexes (and complexes) are often
\emph{chromatic}, since each vertex $v$ of a simplex is  labeled with
a distinct process identity.  

\begin{definition}[Task]
A \emph{task} $T$ for $n$ processes is a triple
$(\m{I},\m{O},\Delta)$ where $\m{I}$ and $\m{O}$ are pure chromatic $(n-1)$-dimensional complexes, and
 $\Delta$  maps each simplex $s$ from $\m{I}$ to a subcomplex $\Delta(s)$ of
 $\m{O}$, satisfying:
\begin{enumerate}
\item $\Delta(s)$ is pure of dimension $s$
\item For every $t$ in $\Delta(s)$ of dimension $s$, $\ID(t) = \ID(s)$
\item If $s,s'$ are two simplexes in $\m{I}$ with $s'\subset s$ then $\Delta(s')\subset\Delta(s)$.
\end{enumerate}
\end{definition}
We say that $\Delta$ is a \emph{carrier map} from the input complex $\m{I}$
to the output complex $\m{O}$.

A task is a very compact way of specifying a distributed problem,
and indeed it is hard to understand what exactly is the problem being specified.
Intuitively, $\Delta$ specifies, for every
simplex $s\in\m{I}$, the valid outputs $\Delta(s)$ for the  processes in $\ID(s)$
assuming they run to completion, and the other processes crash initially,
and do not take any steps.

The immediate snapshot task is depicted in Figure~\ref{fig:fig-is-task}.
On the left, the input simplex is depicted and, on the right,
the output complex appears.

\begin{figure*}
\begin{center}
	\epsfxsize=5.0in
	\epsfbox{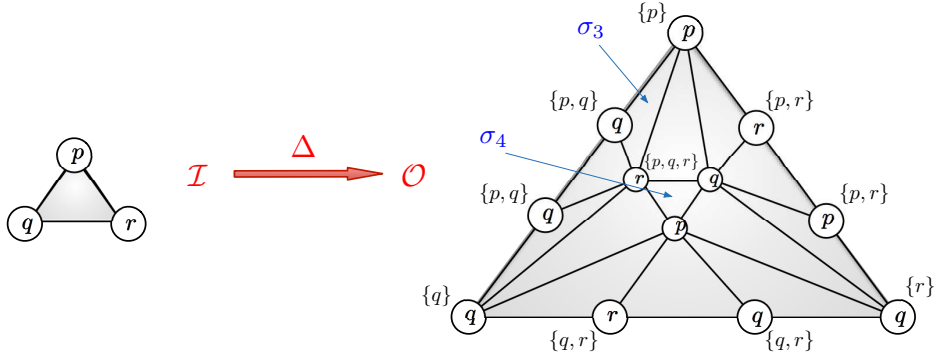}%
    \end{center}
\caption{Immediate snapshot task}
\label{fig:fig-is-task}
\end{figure*}

In figure~\ref{fig:fig-ws-task} one simplex $s$ is added to the output complex
of the immediate snapshot task of Figure~\ref{fig:fig-is-task},
where $s=\{(p,\{p,q\}),(q,\{p,q,r\}),(r,\{p,q,r\}) \}$.
This simplex $s$ corresponds to the execution of Figure~\ref{fig:simple-counterEx2}.
\begin{figure*}
\begin{center}
	\epsfxsize=2.7in
	\epsfbox{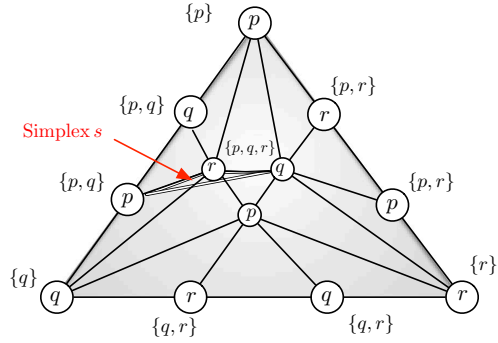}%
    \end{center}
\caption{Part of the write-snapshot output complex}
\label{fig:fig-ws-task}
\end{figure*}

\subsection{Validity as a task}
\label{app:validity}
Recall the validity  object is specified as an interval-sequential object
in Section~\ref{sec:valObj},
which is neither linearizable nor set-linearizable.
In the usual, informal style of specifying a task, the
definition would be very simple:
 an operation returns a value that has been proposed. 
 A bit more formally, in an execution where a set of
 processes participate with inputs $I$ (each $x\in I$ is
 proposed by at least one process), each participating
 process decides a value in $I$.
To illustrate why this informal style can be misleading, consider
the execution in Figure~\ref{fig-validityMislead}, 
where the three processes propose values $I=\{1,2,3\}$,
so according to the informal description it should
be ok that they decide values  $\{1,2,3\}$. However,
 for the detailed interleaving of the figure, it is not
 possible that $p$ and $q$ would have produced outputs
 that they have not yet seen.

To define validity formally as a task, the following notation
will be useful. It defines a complex that represents
all possible assignments of  (not necessarily distinct) values from
a set $U$ to the processes. In particular, all processes can
get the same value $x$, for any $x\in U$.
Given any finite
set $U$ and any integer $n\geq 
1$, we denote by $\complex(U,n)$ the $(n-1)$-dimensional \emph{pseudosphere}~
\cite{HKRbook}
complex induced by $U$: for each  $i\in [n]$ and each $x \in U$, there is a vertex
labeled $(i,x)$ in the vertex set of $\complex(U,n)$. Moreover, $u =
\{(\id_1,u_1),\ldots,(\id_k,u_k)\}$ is a simplex of $\complex(U,n)$
if and only if $u$ is properly colored with identities, that is $\id_i \neq
\id_j$ for every $1 \leq i < j \leq k$. In particular,
$\complex(\{0,1\},n)$ is (topologically equivalent) to the $(n-1)$-dimensional sphere.  For $u\in\complex(U,n)$, we denote by $\val(u)$ 
the set formed of all the values in $U$ corresponding
to the processes in $u$. Similarly, for any set of processes $P$,
 $\complex(U,P)$ is the $|P-1|$-dimensional {pseudosphere}
 where each vertex is labeled with a process in $P$,
 and gets a value from $U$.

The \emph{validity task} over a set of values $U$ that
can be proposed, is $(\m{I},\m{O},\Delta)$, 
where $\m{I}= \m{O}=\complex(U,n)$.
The carrier map $\Delta$ is defined as follows.
For each simplex $s\in\m{I}$, $\Delta(s)=\complex(U',P')$,
where $P'$ is the set of processes appearing in $s$
and $U'$ is their proposed values.

\section{Proofs}
\label{app:completeness}

\begin{claim-repeat}{claim-acyclic}
The relation $\stackrel{S}{\longrightarrow}$ is acyclic.
\end{claim-repeat}

\begin{proofCl}
For the sake of contradiction, suppose that $\stackrel{S}{\longrightarrow}$ 
is not acyclic, namely, there is a cycle 
$C = S_1 \stackrel{S}{\longrightarrow} S_2 \stackrel{S}{\longrightarrow} \hdots
\stackrel{S}{\longrightarrow} S_{m-1} \stackrel{S}{\longrightarrow} S_m$ ,
 with $S_1 = S_m$.
We will show that the existence of $C$ implies that 
$\stackrel{S_X}{\longrightarrow}$ is
not acyclic, for some object $X$, which is a contradiction to our 
initial assumptions.

First note that it cannot be that each $S_i$ is a concurrency class 
of the same object
$X$, because if so then $C$ is a cycle of $\stackrel{S_X}{\longrightarrow}$, 
which contradicts that  $\stackrel{S_X}{\longrightarrow}$ is a total order.
Thus, in $C$ there are concurrency classes of several objects.

In what follows, by slight abuse of notation, we will write
$S' \stackrel{op}{\longrightarrow} S''$ if $S' $ and $S''$ are related
in $\widehat S$ because of the second case in the definition of 
$\stackrel{S}{\longrightarrow}$.

Note that in $C$ there is no sequence 
$S_1 \stackrel{op}{\longrightarrow} S_2 \stackrel{op}{\longrightarrow} S_3$
because in $\widehat S$ whenever $T' \stackrel{op}{\longrightarrow} T''$, 
we have that
$T'$ is a responding class and $T''$ is an invoking class, by definition.
Thus, in $C$ there must be a sequence 
$S_1 \stackrel{op}{\longrightarrow} S_2 \stackrel{S_X}{\longrightarrow} 
\hdots \stackrel{S_X}{\longrightarrow} S_t 
\stackrel{op}{\longrightarrow} S_{t+1} 
\stackrel{S_Y}{\longrightarrow} S_{t+2}$.
Observe that in $\widehat S$, we have $S_2 \stackrel{S_X}{\longrightarrow} S_t$
since $\stackrel{S_X}{\longrightarrow}$ is transitive, hence
the sequence can be shortened: 
$S_1 \stackrel{op}{\longrightarrow} S_2 \stackrel{S_X}{\longrightarrow} S_t 
\stackrel{op}{\longrightarrow} S_{t+1} \stackrel{S_Y}{\longrightarrow} S_{t+2}$.
Note that $S_1$ and $S_t$ are responding classes 
while $S_2$ and $S_{t+1}$ are invoking classes.

Now, since $S_1 \stackrel{op}{\longrightarrow} S_2$, there are operations 
$a, b \in OP$
such that $a \stackrel{op}{\longrightarrow} b$, $a \in S_1$ and $b \in S_2$.
Similarly, for $S_t \stackrel{op}{\longrightarrow} S_{t+1}$, there are
 $c, d \in OP$ such that $c \stackrel{op}{\longrightarrow} d$, $c \in S_t$ 
and $d \in S_{t+1}$.
This implies that $term(a) < init(b)$ and $term(c) < init(d)$.
Observe that if we show $a \stackrel{op}{\longrightarrow} d$ then,
by definition of $\widehat S$, we have
$S_1 \stackrel{op}{\longrightarrow} S_{t+1}$, if $S_1$ and $S_{t+1}$ are 
concurrent classes
of distinct objects, and $S_1 \stackrel{S_Y}{\longrightarrow} S_{t+1}$
otherwise. 
Hence, the sequence 
$S_1 \stackrel{op}{\longrightarrow} S_2 \stackrel{S_X}{\longrightarrow} S_t 
\stackrel{op}{\longrightarrow} S_{t+1} \stackrel{S_Y}{\longrightarrow} S_{t+2}$
can be shortened to 
$S_1 \stackrel{op}{\longrightarrow} S_{t+1} 
\stackrel{S_Y}{\longrightarrow} S_{t+2}$ or 
$S_1 \stackrel{S_Y}{\longrightarrow} S_{t+1} 
\stackrel{S_Y}{\longrightarrow} S_{t+2}$.
Repeating this enough times, in the end 
we get that there are concurrency classes $S_i, S_j$ in $C$ such that 
$S_i \stackrel{S_X}{\longrightarrow} S_j \stackrel{S_X}{\longrightarrow} S_i$,
which is a contradiction since $\stackrel{S_X}{\longrightarrow}$
 is acyclic, by hypothesis.

To complete the proof of the claim, we need to show that 
$a \stackrel{op}{\longrightarrow} d$, i.e., $term(a) < init(d)$.
We have four cases:
\begin{enumerate}
\vspace{-0.2cm}
\item if $b = c$, then $term(a) < init(b) < term(b) = term(c) < init(d)$, hence
$a \stackrel{op}{\longrightarrow} d$.
\vspace{-0.2cm}
\item If $b \stackrel{op}{\longrightarrow} c$, then 
$term(a) < init(b) < term(b) < init(c) < term(c) < init(d)$, hence
$a \stackrel{op}{\longrightarrow} d$.
\vspace{-0.2cm}
\item If $c \stackrel{op}{\longrightarrow} b$, then we have that 
$S_t \stackrel{S_X}{\longrightarrow} S_2$, 
because each $\widehat S|_X = (S_X, \stackrel{S_X}{\longrightarrow} )$
respects the real time order in $comp(\overline E|_X)$, by hypothesis.
But we also have that $S_2 \stackrel{S_X}{\longrightarrow} S_t$,
which contradicts that $\stackrel{S_X}{\longrightarrow}$ is a total order.
Thus this case cannot happen.
\vspace{-0.2cm}
\item If $b$ and $c$ are concurrent, i.e.,  $b \stackrel{op}{\nrightarrow} c$
and $c \stackrel{op}{\nrightarrow} b$, then
note that if $init(d) \leq term(a)$,
then $term(c) < init(d) \leq term(a) < init(b)$,
which implies that $c \stackrel{op}{\longrightarrow} b$ and hence $b$ 
and $c$ are not concurrent,
fro which follows that $term(a) < init(d)$.
\end{enumerate}
\vspace{-0.6cm}
\renewcommand{\toto}{claim-acyclic}
\end{proofCl}

\begin{theorem-repeat}{theo-non-blocking}
Let $E$ be an interval-linearizable execution in which there is a pending invocation $inv(op)$ 
of a total operation. Then, there is a response $res(op)$ such that $E \cdot res(op)$ is interval-linearizable.
\end{theorem-repeat}

\begin{proofT}
Since $E$ is interval-linearizable, there is an interval-linearization $\widehat S \in ISS(X)$ of it.
If $inv(op)$ appears in $\widehat S$, we are done, because $\widehat S$ contains only 
completed operations and actually  it is an interval-linearization of $E \cdot res(op)$,
where $res(op)$ is the response to $inv(op)$ in $\widehat S$.

Otherwise, since the operation is total, there is a responding concurrency class $S'$
such that $\widehat S \, \cdot \, \{inv(op)\} \, \cdot \, S' \in ISS(X)$,
which is an interval-linearization of $E \cdot res(op)$,
where $res(op)$ is the response in $S'$ matching $inv(op)$.
\renewcommand{\toto}{theo-non-blocking}
\end{proofT}

\begin{figure}[ht]
\centering{
\fbox{
\begin{minipage}[t]{150mm}
\footnotesize
\renewcommand{\baselinestretch}{2.5}
\resetline
\begin{tabbing}
aa\=aaaa\=aaaa\=aaaa\=aaaa\=aaaa\=aaaa\=aaaa\=\kill 

{\bf function} ${\sf sequences}$ ($E$) {\bf is} \\

\> $i \leftarrow 1$;
$e \leftarrow$ first event in $E$;
$F \leftarrow$ empty execution\\

\> $\sigma_0, \tau_0 \leftarrow \emptyset$;
$A \leftarrow \sigma_0$; $B \leftarrow \tau_0$;\\

\> \textbf{while} $F \neq E$ \textbf{do}\\

\> \> \textbf{if} \= ($e$ is an invocation) \\

\> \> \> \textbf{then} \=
 \textbf{if}  \=  ($\tau_i \setminus \tau_{i-1} = \emptyset$)\\

\> \> \>  \> \> 
 \textbf{then} \=  $\sigma_i \leftarrow \sigma_i \cup \{ e \}$\\

\> \> \>  \> \>   \textbf{else} \>
 $\sigma_{i+1} \leftarrow \sigma_i \cup \{ e \}$;
         $\tau_{i+1} \leftarrow \tau_i$;\\

\> \> \> \> \> \>   
 $A \leftarrow A \cdot \sigma_i$; $B \leftarrow B \cdot \tau_i$;
          $i \leftarrow i+1$\\

\> \> \> \> \textbf{end if}\\

\>\>\> \textbf{else} \>  $\tau_i \leftarrow \tau_i \cup \{ e \}$\\

\>\> \textbf{end if};\\

\> \>$F \leftarrow F \cdot e$; $e \leftarrow $ next event to $e$ in $E$\\

\> \textbf{end while};\\

\> $A \leftarrow A \cdot \sigma_i$; $B \leftarrow B \cdot \tau_i$;

 $\textsf{return}$ $(A, B)$.

\end{tabbing}
\normalsize
\end{minipage}
}
\caption{Producing a sequence of faces of the invocation simplex $\sigma$ and 
response simplex $\tau$ of an execution $E$.} 
\label{fig:sequence-faces}
}
\end{figure}

For the rest of the section we will often use the following notation.
Let $E$ be an execution. Then, $\sigma_E$ and $\tau_E$ denote
the sets containing all invocations and responses of $E$, respectively.

\begin{claim}
\label{claim-sequence}
For every execution $E$, the function ${\sf sequences}()$ in 
Figure~{\em \ref{fig:sequence-faces}} produces sequences 
$A = \sigma_0, \hdots, \sigma_k$ and $B = \tau_0, \hdots, \tau_m$ such that
\begin{enumerate}
\vspace{-0.2cm}
\item $k = m$.
\vspace{-0.2cm}
\item
$\sigma_0 = \emptyset \subset \sigma_1 \subset \hdots \subset \sigma_{m-1} \subset \sigma_m = \sigma_E$
and
$\tau_0 = \emptyset \subset \tau_1 \subset \hdots \subset \tau_{m-1} \subseteq \tau_m = \tau_E$.
\vspace{-0.2cm}
\item If $E$ has no pending invocations, then $\tau_{m-1} \subset \tau_m$,
otherwise $\tau_{m-1} = \tau_m$
\vspace{-0.2cm}
\item If $E$ satisfies a task with carrier map $\Delta$, then, that for each $i$, 
$\tau_i \in \Delta(\sigma_i)$.
\vspace{-0.2cm}
\item For every response $e$ and invocation $e'$ in $E$ such that
$e$ precedes $e'$ and they do not match each other, 
we have $i < j$, where $i$ is the smallest integer such that 
$e \in \tau_i$ and $j$ is the smallest integer such that $e' \in \sigma_j$.
\vspace{-0.2cm}
\item For every response $e$ of $E$ in $\tau_i \setminus \tau_{i-1}$,
$\sigma_i$ contains all invocations preceding $e$ in $E$.

\end{enumerate}
\end{claim}

\begin{proofCl}
Items (1), (2) and (6) follow directly from the code. 
For item (3), note that if $E$ has no pending invocations, 
it necessarily ends with a response,
which is added to $\tau_m$, and thus $\tau_{m-1} \subset \tau_m$.
For item (4), consider a pair $\sigma_i$ and $\tau_i$, 
and let $E'$ be the shortest prefix of $E$ that contains
each response in $\tau_i$. Note that the simplex containing all 
invocations in $E'$ is $\sigma_i$.
Since, by hypothesis, $E$ satisfy $T$,
it follows that $\tau_i \in \Delta(\sigma_i)$.

For item (5), consider such events $e$ and $e'$.
Since $e$ precedes $e'$, the procedures analyzes
first $e$, from which follows that it is necessarily true that $i \leq j$,
so we just need to prove that $i \neq j$.
Suppose, by contradiction, that $i = j$.
Consider the beginning of the while loop when $e'$ is analyzed. 
Note that $e' \notin \sigma_i$ at that moment.
Also, note that $e \in \tau_i \setminus \tau_{i-1}$
because $i$ is the smallest integer such that $e \in \tau_i$
and its was analyzed before $e'$. Thus,
when the procedure process $e'$, puts it in $\sigma_{i+1}$, 
which is a contradiction,
because in the final sequence of simplexes $e' \notin \sigma_i$.
\renewcommand{\toto}{claim-sequence}
\end{proofCl}

\begin{theorem-repeat}{theo-from-tasks-to-objects}
For every task $T$, there is an interval-sequential object $O_T$ such that 
any execution $E$ satisfies $T$ if and only if 
it is interval-linearizable with respect to $O_T$.
\end{theorem-repeat}

\begin{proofT}
The structure of the proof is the following. 
(1) First, we define $O_T$ using $T$,
(2) then, we show that every execution that satisfies $T$,
is interval-linearizable with respect to $O_T$, and
(3) finally, we prove that every execution that is interval-linearizable
with respect to $O_T$, satisfies $T$. 

\paragraph{Defining $O_T$:}
Let $T = \langle \m{I}, \m{O}, \Delta \rangle$.
To define $O_T$,  we first define its sets with invocations, response and states:
$Inv = \{ (id, x) : \{(id, x)\} \in \m{I} \ \}$,
$Res = \{ (id, y) : \{(id, y)\} \in \m{O} \ \}$ and 
$Q = \{ (\sigma, \tau) : \sigma \in \m{I} \wedge \tau \in \m{O} \}$.
The interval-sequential object $O_T$ has one initial state: $(\emptyset, \emptyset)$.
Then $O_T$ will have only one operation and so the name of it does not appear
in the invocation and responses. 

The transition function $\delta$ is defined as follows.
Consider an input simplex $\sigma$ of $T$ and let $E$ be an execution that satisfies
$T$ with $\sigma_E = \sigma$ and $\tau_E \in \Delta(\sigma_E).$
Consider the sequences 
$\sigma_0 = \emptyset \subset \sigma_1 \subset \hdots \subset \sigma_m = \sigma_E$
and 
$\tau_0 = \emptyset \subset \tau_1 \subset \hdots \subset \tau_m = \tau_E$
that {\sf Sequences} in Figure~\ref{fig:sequence-faces} produces on $E$.
Then, for every $i = 1, \hdots, m$,  
$\delta((\sigma_{i-1}, \tau_{i-1}), \sigma_i \setminus \sigma_{i-1})$ contains
$((\sigma_i, \tau_i), \tau_i \setminus \tau_{i-1})$.
In other words, we use the sequences of faces to define 
an interval-sequential execution (informally, a grid) that will be accepted by $O_T$:
the execution has $2m$ concurrency classes, 
and for each $i = 1, \hdots, m$, the invocation $(p_j, -)$ (of process $p_j$) belongs to 
the $2i-1$-th concurrency class if $(p_j,-) \in \sigma_i \setminus \sigma_{i-1}$,
and the response to the invocation of appears in the $2i$-th concurrency class
if $(p_j, -) \in \tau_i \setminus \tau_{i-1}$.
We repeat the previos construction for every such $\sigma$ and $E$.

\paragraph{If $E$ satisfies $T$, it is interval-linearizable:}
Consider an execution $E$  that satisfies $T$. 
We prove that $E$ is interval-linearizable with respect to $O_T$.
Since $E$ satisfies $T$, we have that $\tau_E \in \Delta(\sigma_E)$.
By definition, $\Delta(\sigma)$ is $\dim(\sigma)$-dimensional pure,
then there is a $\dim(\sigma)$-dimensional  $\gamma \in \Delta(\sigma)$
such that $\tau_E$ is a face of $\gamma$.
Let $\overline E$ be an extension of $E$ in which the responses in $\gamma \setminus \tau_E$
are added in some order. Thus, there are no pending operation in~$\overline E$.
Consider the sequences of simplexes produced by {\sf Sequences} in Figure~\ref{fig:sequence-faces}
on $\overline E$. As we did when defined $O_T$,
the two sequence define an interval-sequential execution $\widehat S$.
We have the following:
(1) $\widehat S$ is an interval-sequential execution of $O_T$, by construction,
(2) for every $p$, $\widehat S |_p = \overline E |_p$, by construction, and
(3) Claim~\ref{claim-sequence}.5 implies that $\widehat S$
respect the real-time order of invocations and responses in $\overline E$: if
$op_1 \stackrel{op}{\longrightarrow} op_2$ in the partial order 
$\widehat OP = (OP, \stackrel{op}{\longrightarrow} )$ associated to $\overline E$,
then, by the claim, the response of $op_1$ appears for the first time the sequence in $\tau_i$
and the invocation of $op_2$ appears for the first in the sequence in $\sigma_j$,
with $i < j$, and hence, by construction, $op_1$ precedes $op_2$ in $\widehat S$.
We conclude that $\widehat S$ is an  interval-linearization of~$E$.

\paragraph{If $E$ is interval-linearizable, it satisfies $T$:}
Consider an execution $E$ that is interval-linearizable
with respect to $O_T$.  We will show that $E$ satisfies $T$.
There is an interval-sequential execution $\widehat S$ that is a linearization of $E$,
since $E$ is interval-linearizable.
Consider any prefix $E'$ of $E$ and let $\widehat S'$ be the shortest prefix of $\widehat S$
such that (1) it is an interval-sequential execution and (2) every completed invocation in $E'$ 
is completed in $\widehat S'$ (note that there might be pending invocations in $E'$ that does not
appear in $\widehat S'$). By construction, $\widehat S'$ defines two sequences of 
simplexes $\sigma_0 = \emptyset \subset \sigma_1 \subset \hdots \subset \sigma_m$
and 
$\tau_0 = \emptyset \subset \tau_1 \subset \hdots \subset \tau_m$
with $\tau_i \in \Delta(\sigma_i)$, for every $i = 1, \hdots, m$.
Observe that $\sigma_{E'} = \sigma_m$ and $\tau_{E'} \subseteq \tau_m$,
and thus $\tau_{E'} \in \Delta(\sigma_{E'})$ because $\tau_m \in \Delta(\sigma_m)$.
What we have proved holds for every prefix $E'$ of $E$,
then we conclude that $E$ satisfies $T$.
\renewcommand{\toto}{theo-from-tasks-to-objects}
\end{proofT}

\begin{lemma-repeat}{lemm-queueNoTask}
There is a sequential one-shot object $O$ such that there is no
task $T_O$, satisfying that an execution $E$ is linearizable with respect to $O$
if and only if $E$ satisfies $T_O$ (for every $E$).
\end{lemma-repeat}

\begin{proofL}
Consider a restricted queue $O$ for three processes, $p$, $q$ and $r$,
in which, in every execution, $p$ and $q$ invoke $enq(1)$ and $enq(2)$, respectively,
and $r$ invokes $deq()$. If the queue is empty, $r$'s dequeue operation gets $\bot$.

Suppose, for the sake of contradiction, that there is a corresponding task 
$T_O=(\m{I},\m{O},\Delta)$, as required by the lemma.
The input complex $\m{I}$ consists of one vertex for
each possible operation by a process, namely, the set of
vertices is $\{(p,enq(1)),(q,enq(2)),(r,deq())\}$, and $\m{I}$ consists
of all subsets of this set. Similarly, the output complex $\m{O}$
contains one vertex for every possible response to a process,
so it consists of the set of vertices 
$\{(p,ok),(q,ok),(r,1),(r,2),(r,\bot)\}$. It should contain
a simplex $\sigma_x=\{(p,ok),(q,ok),(r,x)\}$ for each value of $x\in\{1,2,\bot\}$,
because there are executions where  $p,q,r$ get such values, respectively.
See Figure~\ref{fig:fig-queueCounterE}.
\begin{figure*}
\begin{center}
	\epsfxsize=6.0in
	\epsfbox{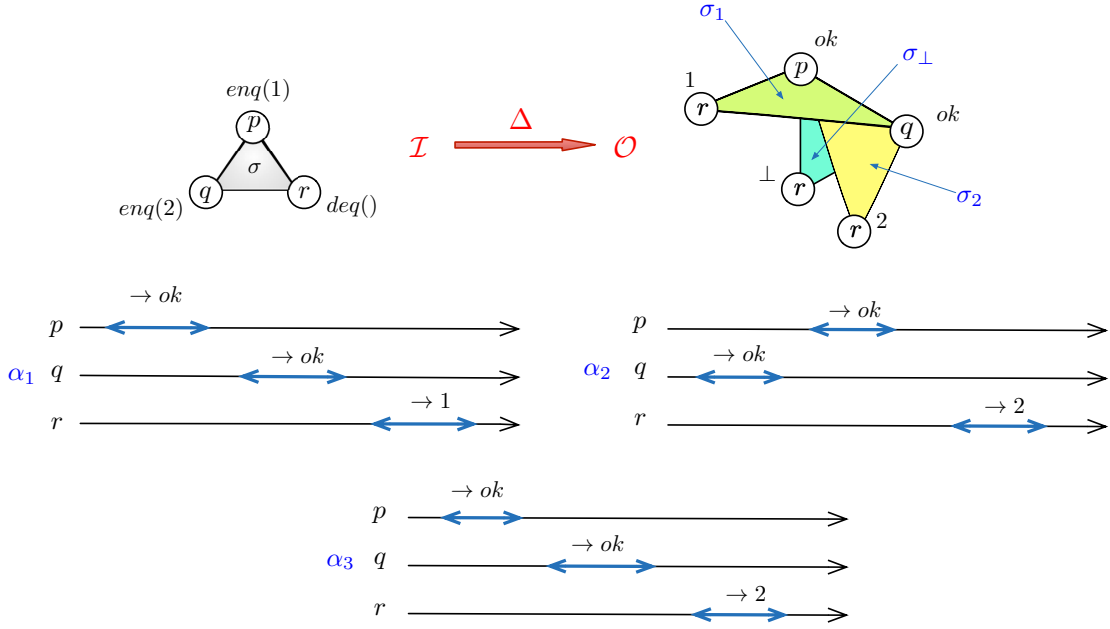}%
    \end{center}
\caption{Counterexample for a simple queue object}
\label{fig:fig-queueCounterE}
\end{figure*}

Now,  consider the  three sequential executions of the figure, $\alpha_1,\alpha_2$
and~$\alpha_\bot$. In $\alpha_1$ the process execute their operations
in the order $p,q,r$, while in $\alpha_2$ the order is $q,p,r$.
In $\alpha_1$ the response to $r$ is $1$, and if $\alpha_2$ it is $2$.
Given that these executions are linearizable for $O$,
they should be valid for $T_O$. This means that every  prefix of $\alpha_1$
 should be valid: 
\begin{align*}
\{ (p,ok)\} &= \Delta((p,enq(1)) \\ 
\{ (p,ok),(q,ok)\} &\in \Delta(\{ (p,enq(1),(q,enq(2)\}) \\
\sigma_1=\{ (p,ok),(q,ok),(r,1)\} &\in \Delta(\{ (p,enq(1),(q,enq(2),(r,deq())\})=\Delta(\sigma)
\end{align*}
Similarly from $\alpha_2$ we get that
$$
\sigma_2=\{ (p,ok),(q,ok),(r,2)\} \in \Delta(\sigma)
$$
But now consider $\alpha_3$, with the same sequential  order  $p,q,r$
of operations, but now $r$ gets back value $2$.
This execution is not linearizable for $O$, but is accepted by $T_O$
because each of the prefixes of $\alpha_3$ is valid.
More precisely,   the set of inputs and the set of outputs of $\alpha_2$
are identical to the sets of inputs and set of outputs of $\alpha_3$.
\renewcommand{\toto}{lemm-queueNoTask}
\end{proofL}

\begin{theorem-repeat}{theo-from-objects-to-tasks}
For every one-shot interval-sequential object $O$ with a single total operation,
there is a refined task $T_O$ such that any execution $E$ is interval-linearizable with respect to $O$
if and only if $E$ satisfies $T_O$.
\end{theorem-repeat}

\begin{proofL}
The structure of the proof is the following. 
(1) First, we define $T_O$ using $O$,
(2) then, we show that every execution that is interval-linearizable
with respect to $O$, satisfies $T_O$, and
(3) finally, we prove that every execution that satisfies $T_O$,
is interval-linearizable with respect to $O$.

\paragraph{Defining $T_O$:}
We define the refined task $T_O = (\m{I}, \m{O}, \Delta)$.
First, since $O$ has only one operation, we assume that its invocation
and responses have the form $inv(p_j, x_j)$ and $res(p_j, y_j)$.
Let $Inv$ and $Res$ be the sets with the invocations and responses of $O$.
Each subset $\sigma$ of $Inv$ containing invocations with different processes,
is a simplex of $\m{I}$. It is not hard to see that $\m{I}$ is a chromatic pure $(n-1)$-dimensional complex.
$\m{O}$ is defined similarly: each subset $\tau$ of $Res \times 2^{Inv}$ containing responses in the first entry 
with distinct processes, is a simplex of $\m{O}$.

We often use the following construction in the rest of the proof.
Let $E$ be an execution. 
Recall that $\sigma_E$ and $\tau_E$ are the simplexes (sets)
containing all invocations and responses in $E$.
Let $\tau_0 = \emptyset \subset \tau_1 \subset \hdots \subset \tau_m = \tau_E$
and 
$\sigma_0 = \emptyset \subset \sigma_1 \subset \hdots \subset \sigma_m = \sigma_E$
be the sequences of simplexes produced by {\sf Sequences} in Figure~\ref{fig:sequence-faces}
on $E$. 
These sequences define an output simplex $\{ (res, \sigma_i) :  res \in \tau_i \setminus \tau_{i-1}\}$ of $T_O$,
which will be denoted $\gamma_E$.

We define $\Delta$ as follows.
Consider an input simplex $\sigma \in \m{I}$.
Suppose that $E$ is an execution such that 
(1) $\sigma_E = \sigma$
(2) it has no pending invocations and
(3) it is interval-linearizable with respect to $O$.
Note that $\dim(\sigma) = \dim(\tau_E)$ and $ID(\sigma) = ID(\tau_E)$.
Consider the simplex $\gamma_E$ induced by $E$, as defined above.
Note that  $\dim(\sigma) = \dim(\gamma_E)$ and $ID(\sigma) = ID(\gamma_E)$.
Then, $\Delta(\sigma)$ contains $\gamma_E$ and all its faces.
We define $\Delta$ by repeating this for each such $\sigma$ and $E$.

Before proving that $T$ is a task, we observe the following.
Every $\dim(\sigma)$-simplex $\gamma_E \in \Delta(\sigma)$
is induced by an execution $E$ that is interval-linearizable.
Let $\widehat S$ be an interval-linearization of $E$.
Then, for any execution $F$ such that $\gamma_F = \gamma_E$ (namely, {\sf Sequences} produce
the same sequences on $E$ and on $F$), $\widehat S$ is an interval-linearization of $F$:
(1) by definition, $\widehat S$ is an interval-sequential execution of $O$,
(2) for every process $p$, $F|_p = \widehat S|_p$, and
(3) $\widehat S$ respects the real-time order of invocations and response in $F$
because it respects that order in $E$ and, as already said, they have the same sequence 
of simplexes and, by Claim~\ref{claim-sequence}.5, 
these sequences reflect the real-time order of invocations and responses.

We argue that $T$ is a refined task.
Clearly, for each $\sigma \in \m{I}$, $\Delta(\sigma)$ is a pure and chromatic complex of dimension $\dim(\sigma)$,
and for each $\dim(\sigma)$-dimensional $\tau \in \Delta(\sigma)$, $ID(\tau) = ID(\sigma)$.
Consider now a proper face $\sigma'$ of $\sigma$.
By definition, each $\dim(\sigma')$-dimensional simplex $\gamma_{E'} \in \Delta(\sigma')$
corresponds to an execution $E'$ whose set of invocations is $\sigma'$,
has no pending invocations and is interval-linearizable with respect to $O$.
Let $\widehat S'$ be an interval-linearization of $E'$ and let $s'$
be the state $O$ reaches after running $\widehat S'$.
Since the operation of $O$ is total, from the state $s'$, the invocations in $\sigma \setminus \sigma'$ can be invoked one by 
one until $O$ reaches a state $s$ in which all invocations in $\sigma$ have a matching response.
Let $\widehat S$ be the corresponding interval-sequential execution and let 
$E$ be an extension of $E'$ in which (1) every invocation in $\sigma \setminus \sigma'$ has the response 
in $\widehat S$ and (2) all invocations first are appended in some order and then all responses
are appended in some  (not necessarily the same) order.
Note that $\widehat S'$ is a prefix of $\widehat S$ and actually $\widehat S$ is
an interval-linearization of $E$.
Then, $\Delta(\sigma)$ contain the $\dim(\sigma)$-simplex $\gamma_E$ induced by $E$.
Observe that $\gamma_{E'} \subset \gamma_E$, hence, $\Delta(\sigma') \subset \Delta(\sigma)$.
Finally, for every $\dim(\sigma)$-simplex $\gamma_E \in \Delta(\sigma)$, 
for every vertex $(id_i, y_i, \sigma_i)$ of $\gamma_E$, we have that $\sigma_i \subseteq \sigma$
because $\gamma_E$ is defined through an interval-linearizable execution $E$.
Therefore, we conclude that $T_O$ is a task.


\paragraph{If $E$ is interval-linearizable, it satisfies $T_O$:}
Consider an execution $E$ that is interval-linearizable with respect to $O$.
We prove that $E$ satisfies $T_O$.
Since $E$ is interval-linearizable, 
there is an interval-linearization $\widehat S$ of it.
Let $E'$ be any prefix of $E$ and let 
$\widehat S'$ be the shortest prefix of $\widehat S$
such that (1) it is an interval-sequential execution and (2) every completed invocation in $E'$ 
is completed in $\widehat S'$ (note that there might be pending invocations in $E'$ that does not
appear in $\widehat S'$).
Note that $\widehat S'$ is an interval-linearization of $E'$.
Since interval-linearizability is non-blocking, Theorem~\ref{theo-non-blocking}, 
there is an interval-linearization $\widehat S''$ of $E'$ in which every invocation
in $E'$ has a response. Let $E''$ be an extension of $E'$ that contains
every response in $\widehat S''$ (all missing responses are appended at the end in some order).
Note that $\sigma_{E'} = \sigma_{E''}$ and $\tau_{E'} \subseteq \tau_{E''}$.
Consider the output simplexes $\gamma_{E'}$ and $\gamma_{E''}$ induced by $E'$ and $E'$.
Observe that $\gamma_{E'} \subseteq \gamma_{E''}$ (because $E'$ differs from $E''$
in some responses at the end).
By definition of $T_O$, $\gamma_{E''} \in \Delta(\sigma_{E'})$,
and hence $\gamma_{E'} \in \Delta(\sigma_{E'})$.
This holds for every prefix $E'$ of $E$, and thus $E$ satisfies~$T_O$.

\paragraph{If $E$ satisfies $T_O$, it is interval-linearizable:}
Consider an execution $E$ that satisfies $T_O$. We prove that $E$ is interval-linearizable
with respect to~$O$.
Let $\gamma_E$ be the output simplex induced by $E$.
Since $E$ satisfies $T_O$, $\gamma_E \in \Delta(\sigma_E)$.
Since $\Delta(\sigma_E)$ is a pure $\dim(\sigma_E)$-dimensional complex,
there is a $\dim(\sigma)$-simplex $\gamma_{\overline E} \in \Delta(\sigma_E)$ such that 
$\gamma_E \subseteq \gamma_{\overline E}$.
Observe that $\sigma_E = \sigma_{\overline E}$ while $\tau_E \subseteq \tau_{\overline E}$.
Therefore, $\overline E$ is an extension of $E$ in which the responses in 
$\tau_{\overline E} \setminus \tau_E$ are appended in some order at the end of $E$.
Now, by construction, $\gamma_{\overline E} \in \Delta(\sigma_E)$ because 
there is an interval-sequential execution $\widehat S$ of $O$ that is
an interval-linearization of $\overline E$.
We have that $\widehat S$ is an interval-linearization of $E$ as well because,
as already mentioned, $\overline E$ is an extension of $E$
in which the responses in 
$\tau_{\overline E} \setminus \tau_E$ are appended in some order at the end.

\renewcommand{\toto}{theo-from-objects-to-tasks}
\end{proofL}

\end{document}